\documentclass[conference]{IEEEtran}
\IEEEoverridecommandlockouts                              

\usepackage{amssymb}
\usepackage{graphicx}
\usepackage{epstopdf}
\usepackage{amsmath}
\usepackage{mathtools, cuted}
\usepackage{multirow}
\usepackage[caption=false,listofformat=subsimple,labelformat=simple]{subfig}

\usepackage{array}
\usepackage{bm}
\usepackage{url}

\newenvironment{proof}{\hspace{12pt}\textit{Proof:}}{\hfill$\blacksquare$}

\newtheorem{problem}{Problem}
\newtheorem{definition}{Definition}
\newtheorem{assumption}{Assumption}
\newtheorem{lemma}{Lemma}
\newtheorem{corollary}{Corollary}
\newtheorem{proposition}{Proposition}

\usepackage[noadjust]{cite}

\title{\LARGE RSSI-based Localization with Adaptive Noise Covariance Estimation\\ for Resilient Multi-Agent Formations}

\author{Paul J Bonczek and Nicola Bezzo
\thanks{Paul J Bonczek and Nicola Bezzo are with the Charles L. Brown Department of Electrical \& Computer Engineering and members of the Link Lab at The University of Virginia, Charlottesville, VA, USA. Email: {\tt \{pjb4xn,nb6be\}@virginia.edu}; {\tt Paul.Bonczek@jhuapl.edu}}
}

\newcommand*{\N}{\mathbb{N}}
\newcommand*{\R}{\mathbb{R}}
\newcommand*{\E}{\mathbb{E}}
\newcommand*{\V}{\mathbb{V}}
\newcommand*{\PP}{\mathbb{P}}
\newcommand{\Tau}{\mathrm{T}}

\renewcommand{\baselinestretch}{1.001}

\begin{document}

\maketitle
\thispagestyle{empty}
\pagestyle{empty}
\begin{abstract}

Typical cooperative multi-agent systems (MASs) exchange information to coordinate their motion in proximity-based control consensus schemes to complete a common objective. However, in the event of faults or cyber attacks to on-board positioning sensors of agents, global control performance may be compromised resulting in a hijacking of the entire MAS. For systems that operate in unknown or landmark-free environments (e.g., open terrain, sea, or air) and also beyond range/proximity sensing of nearby agents, compromised agents lose localization capabilities. To maintain resilience in these scenarios, we propose a method to recover compromised agents by utilizing Received Signal Strength Indication (RSSI) from nearby agents (i.e., mobile landmarks) to provide reliable position measurements for localization. To minimize estimation error: i) a multilateration scheme is proposed to leverage RSSI and position information received from neighboring agents as mobile landmarks and ii) a Kalman filtering method adaptively updates the unknown RSSI-based position measurement covariance matrix at runtime that is robust to unreliable state estimates. The proposed framework is demonstrated with simulations on MAS formations in the presence of faults and cyber attacks to on-board position sensors.

\end{abstract}
\section{Introduction} \label{sec:introduction}

One of the essential capabilities for a mobile autonomous system is to localize itself within an environment. The ability to perform an accurate and robust localization allows unmanned systems to achieve truly autonomous operations. These operations can be accomplished in various ways, by relying on positioning sensors like global positioning system (GPS), odometry, and initial measurement unit or through the use of range sensors such as LiDAR, infrared, and camera systems. The sensing information can then be leveraged via localization methods such as Particle filters and Simultaneous Localization and Mapping (SLAM) techniques. 

When considering multi-agent system (MAS) applications, for example robotic swarms, consensus algorithms are typically considered where agents share their states to attain coordinated behaviors in a decentralized fashion to accomplish a desired goal \cite{formation_control1}. When information being exchanged is incorrect, the MAS can be hijacked and lead to unsafe conditions \cite{7822915}. A variety of issues can cause undesirable information to be exchanged between agents, such as cyber attacks or faults to on-board sensors. If known landmarks are present in the operating space, range sensors can be utilized for localization or to determine if the system is performing as expected. However, landmarks may not be available if an MAS is navigating in an open terrain, thus leaving compromised agents unable to reliably localize themselves.

With such premises, the focus on this paper resides on the problem of resilient coordination in MASs that leverage control consensus schemes when one or more agents lose localization capabilities as on-board positioning sensors become unreliable. More specifically, agents are tasked to maintain desired formations within open or unknown environments that do not offer identifiable landmarks and also operate beyond sensing range of nearby agents. To deal with this, the proposed framework enables compromised agents to leverage Received Signal Strength Indication (RSSI) and received position information from neighboring agents for localization (i.e., mobile landmarks). Multilateration is performed using the noisy RSSI measurements and received neighboring agent's positions to provide RSSI-based position measurements in replacement of the unreliable on-board position sensors. To minimize the RSSI-based position measurement error, a weighted least squares method is used. Moreover, the RSSI-based position measurements have an unknown covariance which differ from the nominal on-board position sensor. To improve state estimation performance, compromised agents leverage an adaptive Kalman Filtering method that estimates the RSSI-based position measurement covariance matrix at runtime. The proposed framework is introduced in a generalized manner that may be used on any formation control technique for swarms of homogeneous linear time-invariant (LTI) modeled agents. As a specific case study in this paper, a virtual spring-damper physics model \cite{Paul_TRO} for proximity-based formation control is considered on MASs in a $2$-dimensional coordinate frame. However, our MAS framework can be expanded to heterogeneous systems \cite{hetergeneous_MAS}, non-linear modeled agents \cite{Nonlinear_MAS}, and higher dimension coordinate frames \cite{positioning_3D}. 
\vspace{-2pt}

\subsection{Related Literature}

The topic of system resilience has received significant attention recently, notably in the area of localization within single- and multi-agent system operations \cite{MAS_survey}. Various time-based measurement techniques used for localization have been proposed, such as Time of Arrival (ToA), Time Difference on Arrival (TDoA), Angle of Arrival (AoA), and Time of Flight (ToF) \cite{localization_survey}. Another technique to obtain ranges are from measuring the Received Signal Strength Indicator (RSSI). As the name suggests, RSSI-based techniques rely on measuring the strength of the received radio frequency signal over the communication channel. 

Much effort has been placed on leveraging RSSI to aid in localization within an environment. In particular, numerous articles have leveraged anchor nodes with known positions to aid in localization within indoor environments \cite{survey_LocalizationRSSI}. One example is found in \cite{hsu2016particle}, the proposed framework utilizes a particle filter on the RSSI measurements from known anchor nodes and then fused the position estimate with the remaining system states for improved localization capabilities. In \cite{RSSI_DVhop}, the Weighted Distance Vector Hop algorithm using RSSI is combined with a weighted hyperbolic localization algorithm to estimate the position of any nearby agents in MASs while utilizing the known locations of anchor nodes in the environment. Authors in \cite{RSSI_NarrowCorridor} proposed a method for robotic swarms deployed in indoor environments to effectively navigate through narrow passageways by allocating specific roles to robots to ensure localization accuracy. An approach \cite{RSSI_mixed} was proposed to provide robustness in localization performance within swarms when nearby agents satisfy both Line-of-Sight and Non-Line-of-Sight conditions. In \cite{Mostofi}, RSSI signals are leveraged to estimate AoA of signal sources (i.e., transmitters) and humans/robots for target tracking. 

Similar to our work, authors in \cite{oliveira2014rssi} proposed an RSSI-based localization algorithm for multi-robot teams within anchor-less environments. Their approach combines a Kalman Filter and the Floyd-Warshall algorithm to compute smooth distance estimates between agents, then multidimensional scaling is utilized to estimate relative positions of nearby agents. Differing from \cite{oliveira2014rssi}, we assume that each agent has a nominal localization sensor that is vulnerable to cyber attacks and faults. As such scenarios arise, our framework allows for any compromised agent to perform sensor reconfiguration by removing the nominal position sensor in favor of RSSI-based position sensing capabilities. Moreover, our decentralized approach does not suffer from scalability issues as in \cite{oliveira2014rssi} where authors claim their framework is effective on mobile robot teams of ``approximately up to 10". 

Our work utilizes a similar principal to previous literature that have characterized adaptive adjustments to estimate unknown noise covariances in dynamic systems \cite{AdaptiveR_EKF}, \cite{AdaptiveR_UKF}. In these works, measurement residuals are used to adaptively update estimates of the noise covariance matrices. However, the authors assumed that measurement noise always follows a zero-mean Gaussian distribution (i.e., attack/fault-free conditions), hence the measurement residuals used for noise covariance estimation are also zero-mean Gaussian random variables. In the presence of compromised sensor measurements, state estimates are unreliable; thus compromising the covariance matrix estimation process.

The contributions of this work include: i) a decentralized framework to provide resiliency to MAS formations in the presence of cyber attacks and faults to critical on-board positioning sensors when operations occur in open/unknown environments that lack known landmarks and agents are beyond range sensing of other neighboring agents. Compromised agents perform sensor reconfiguration to leverage RSSI from the nearby agents (i.e., mobile landmarks) to aid in re-localization in replacement of its compromised on-board position sensor, and ii) a novel adaptive Kalman Filtering approach to update the unknown RSSI-based position measurement covariance at runtime that is robust to unreliable state estimates (improving upon \cite{AdaptiveR_EKF,AdaptiveR_UKF}), to optimize position estimation of compromised agents.

\section{Preliminaries} \label{sec:preliminaries}

\subsection{System Model} \label{sec:System_model}

Let us consider a graph $\mathcal{G} = (\mathcal{V}, \mathcal{E})$ where we denote $\mathcal{V}$ as the set of $N_{\mathrm{a}}$ mobile agents and the set $\mathcal{E} \subset \mathcal{V} \times \mathcal{V}$ defines edge connections between agents. Each agent $i = 1,2,\dots,N_{\mathrm{a}}$ is assumed to have dynamics that can be represented in a discrete-time LTI state space form:
\begin{align} \label{eq:dynamical_model}
    \bm{\mathrm{x}}_{i}^{(k+1)} &= \bm{\mathrm{A}} \bm{\mathrm{x}}_i^{(k)} + \bm{\mathrm{B}} \bm{\mathrm{u}}_{i}^{(k)} + \bm{\nu}_i^{(k)} \\
    \label{eq:output_vector}
    \bm{\mathrm{y}}_i^{(k)} &= \bm{\mathrm{Cx}}_i^{(k)} + \bm{\eta}_i^{(k)}
\end{align}
with state $\bm{\mathrm{A}}$, input $\bm{\mathrm{B}}$, and output $\bm{\mathrm{C}}$ matrices consisting of appropriate dimensions, the state vector $\bm{\mathrm{x}}_i^{(k)} \hspace{-.3pt} \in \hspace{-.3pt} \R^n$, control input $\bm{\mathrm{u}}_{i}^{(k)} \hspace{-.3pt} \in \hspace{-.3pt} \R^{N_{\mathrm{m}}}$, and output vector $\bm{\mathrm{y}}_i^{(k)} \hspace{-.3pt} \in \hspace{-.3pt} \R^{N_{\mathrm{s}}}$. Within the state vector are the position coordinates $\bm{\mathrm{p}}_i^{(k)} \hspace{-.2pt} \in \hspace{-.3pt} \R^D$ in a $D$-dimensional Euclidean space. Process and measurement noises are described as i.i.d. zero-mean Gaussian distributions $\bm{\nu}_i^{(k)} \hspace{-.2pt} \sim \hspace{-.2pt} \mathcal{N}(\bm{0},\bm{\mathrm{Q}}) \hspace{-.2pt} \in \hspace{-.2pt} \R^{n}$ and $\bm{\eta}_i^{(k)} \hspace{-.3pt} \sim \hspace{-.2pt} \mathcal{N}(\bm{0},\bm{\mathrm{R}}) \hspace{-.2pt} \in \hspace{-.2pt} \R^{N_{\mathrm{s}}}$ with covariance matrices $\bm{\mathrm{Q}} > 0$ and $\bm{\mathrm{R}} > 0$, respectively. The use of a Kalman Filter (KF) provides state estimates $\hat{\bm{\mathrm{x}}}_i^{(k|k)} \in \R^n$ and predictions $\hat{\bm{\mathrm{x}}}_i^{(k+1|k)} \in \R^n$ for each agent $i$. 

In order for the multi-agent swarm to cooperatively maintain a desired proximity-based formation, the agents exchange their state estimate information. Each $i$th agent follows a control consensus $\mathcal{U}(\cdot,\cdot,\cdot)$ by
\begin{equation} \label{eq:consensus_control}
\bm{\mathrm{u}}_i^{(k)} = \mathcal{U} \big( \hat{\bm{\mathrm{x}}}_i^{(k|k)} , \hat{\bm{\mathrm{x}}}_j^{(k|k)}, \bm{\mathrm{x}}_{\mathrm{ref}}^{(k)} \big)
\end{equation}
to maintain a proximity-based formation while navigating within an environment, given $i \ne j$ and $j \in \mathcal{S}_i$ where $\mathcal{S}_i \subset \mathcal{V}$ is the neighbor set used for control purposes by an agent $i$ and $\bm{\mathrm{x}}_{\mathrm{ref}}^{(k)}$ is a reference state to follow.

\begin{definition}[Control Graph \cite{Paul_TRO}] \label{def:control_graph}
    Given each agent $i$ in the set $\mathcal V$ having a neighbor set for control $\mathcal{S}_i \subset \mathcal V$, we define the graph $\mathcal{G}_{\mathcal{U}} = ( \mathcal{V}, \mathcal{E}_{\mathcal{U}} )$ with the edge set,
    \begin{equation} \label{eq:control_edges}
        \mathcal{E}_{\mathcal{U}} = \big\{ (i,j) \; \big| \; j \in \mathcal{S}_i, \forall i \in \mathcal{V} \big\}
    \end{equation}
    as the \textit{control graph} of the agent set $\mathcal{V}$.
\end{definition}

\subsection{Threat Model} \label{sec:Attack_model}

During operations, we assume agents may experience cyber attacks or faults to on-board positioning sensors. With the loss of reliable position sensing, localization within the environment is compromised, thus degradation of the control performance occurs within proximity-based formations. Without loss of generality, the output vector from \eqref{eq:output_vector} is formalized in terms of position measurements $\bm{\mathrm{y}}_{i,[1:D]}^{(k)}$ (i.e., measuring the $D$-dimensional position) and sensor measurements of other states (if applicable) that are deemed non-vulnerable $\bm{\mathrm{y}}_{i,[(D+1):N_{\mathrm{s}}]}^{(k)}$ as $\bm{\mathrm{y}}_{i}^{(k)}~=~\big[ ( \bm{\mathrm{y}}_{i,[1:D]}^{(k)} )^{\mathsf{T}} \; ( \bm{\mathrm{y}}_{i,[(D+1):N_{\mathrm{s}}]}^{(k)} )^{\mathsf{T}} \big]^{\mathsf{T}}$. 

\vspace{2pt}
We denote the compromised position measurement vector due to cyber attacks or sensor faults as
\begin{equation} \label{eq:attacked_sensor}
    \widetilde{\bm{\mathrm{y}}}_{i,[1:D]}^{(k)} = \bm{\mathrm{C}}_{[1:D]} \bm{\mathrm{x}}_i^{(k)} + \bm{\eta}_{i,[1:D]}^{(k)} + \bm{\xi}_{i}^{(k)}
\end{equation}
where the vector $\bm{\xi}_{i}^{(k)} \in \R^D$ represents altered position measurements from the nominal behavior on an agent $i \in \mathcal{V}$ and $\bm{\mathrm{C}}_{[1:D]}$ indicates the first $D$ rows of the output matrix corresponding to position. When $\bm{\xi}_{i}^{(k)} \ne 0$ is satisfied (i.e., $\widetilde{\bm{\mathrm{y}}}_{i,[1:D]}^{(k)} \ne \bm{\mathrm{y}}_{i,[1:D]}^{(k)}$), this indicates the presence of cyber attacks manipulating position measurements or a faulty sensor.

\subsection{Communication Model} \label{sec:Communication_model}

\begin{figure}[b!]
\vspace{-10pt}
\hspace{-3pt} \subfloat[\label{fig:pathloss} ]{\includegraphics[width = 0.238\textwidth]{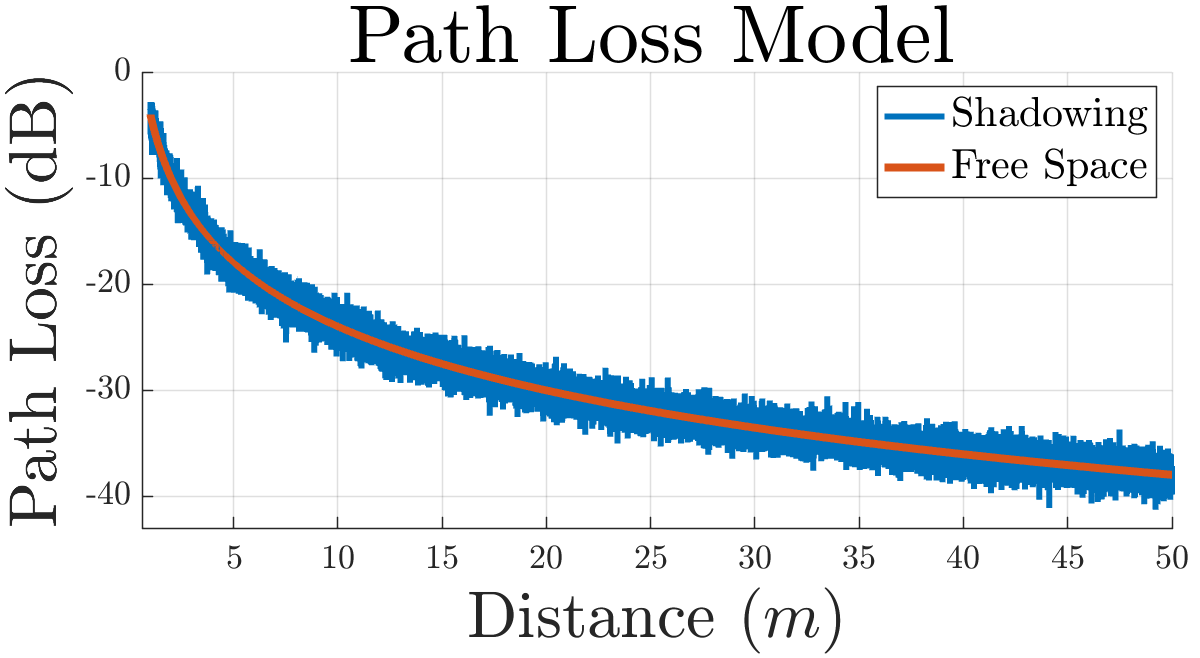}}
\hspace{3pt} \subfloat[\label{fig:estimationError} ]{\includegraphics[width = 0.238\textwidth]{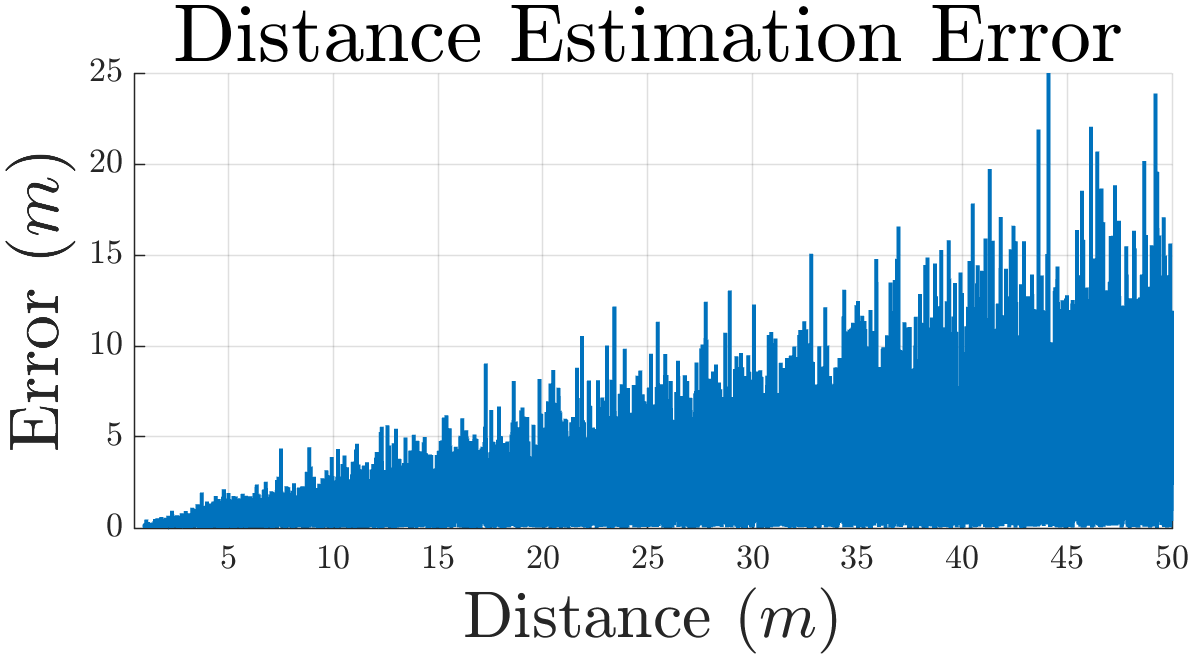}}
\vspace{-1pt}
\caption{An example of the path loss model and the incurred distance estimation error magnitude from RSSI measurements as distance increases.}
\label{fig:PathLossModel}
\vspace{-10pt}
\end{figure}
To overcome malicious cyber attacks or faults to position sensors, agents measure noisy RSSI from the received communications of nearby agents. A commonly-used path loss model is the log-normal shadowing model \cite{goldsmith_2005} defined by:
\vspace{-7pt}
\begin{equation} \label{eq:comm_model}
    P_{ij,[\text{rx}]}^{(k)} = P_{[\text{tx}]} - PL(\mathrm{d}_0) - 10 \beta \log \frac{\mathrm{d}_{ij}^{(k)}}{\mathrm{d}_0} + \Lambda
\end{equation}
where $P_{ij,[\text{rx}]}^{(k)}$ is the measured received power by an agent $i$ of an agent $j$, $PL(\mathrm{d}_0)$ is the power loss (in dB) from a reference distance $\mathrm{d}_0 \in \R_{>0}$, and $\mathrm{d}_{ij}^{(k)} = \| \bm{\mathrm{p}}_i^{(k)} - \bm{\mathrm{p}}_j^{(k)} \|$ denotes the true distance between agents $i$ and $j$. The channel shadowing $\Lambda \sim \mathcal{N}(0,\sigma^2_{\Lambda})$ is modeled as a zero-mean Gaussian noise and $\beta$ is the path loss exponent. It is assumed that all agents have the same transmitting power $P_{[\text{tx}]}$ which is known by the agents beforehand. Fig.~\ref{fig:PathLossModel} provides an example of received signal strength that follows the assumed path loss model with shadowing \eqref{eq:comm_model} and the impact to the corresponding distance estimation as distance between agents increases.

\begin{definition}[Communication Graph \cite{Paul_TRO}] \label{def:comm_graph}
    Given the $N_{\mathrm{a}}$ agents in set $\mathcal V$ with a known maximum communication range $\delta_{\mathrm{c}} \in \R_{>0}$, we define the graph $\mathcal{G}_{\mathcal{C}} = ( \mathcal{V}, \mathcal{E}_{\mathcal{C}} )$ with the edge set represented by
    \begin{equation} \label{eq:comm_edges}
        \mathcal{E}_{\mathcal{C}} = \big\{ (i,j) \; \big| \; \big\| \bm{\mathrm{p}}_i^{(k)} - \bm{\mathrm{p}}_j^{(k)} \big\| \leq \delta_{\mathrm{c}}, \; i,j \in \mathcal{V} \big\}
    \end{equation}
    as the \textit{communication graph} of the agent set $\mathcal{V}$. The set $\mathcal{C}_i = \{ j \in \mathcal{V} \; | \; (i,j) \in \mathcal{E}_{\mathcal{C}} \}$ represents any mobile agent $j$ that is within communication range of an agent $i$ to receive its broadcast information signal.
\end{definition}

\subsection{Problem Formulation} \label{sec:problem}

Given the multi-agent formation topology described by the control graph $\mathcal{G}_{\mathcal{U}}$ that can suffer from degraded control performance due to cyber attacks or faults on individual agent's localization sensors (i.e., $\bm{\xi}_{i}^{(k)} \ne 0$), we are interested in solving the following problems:
\begin{problem} \label{problem1}
\textit{(Detection and Sensor Reconfiguration)} Create a policy such that any agent $i \in \mathcal{V}$ that detects anomalous position sensor measurement behavior can reconfigure its sensor model in the $D$-dimensional space to satisfy:
\begin{equation} \label{eq:reconfiguration}
    \widetilde{\bm{\mathrm{y}}}_{i,[1:D]}^{(k)} \longrightarrow \bar{\bm{\mathrm{y}}}_{i,[1:D]}^{(k)}
\end{equation}
by leveraging the known communication model to provide reliable position measurements $\bar{\bm{\mathrm{y}}}_{i,[1:D]}^{(k)}$ to re-localize itself.
\end{problem}
\vspace{1pt}

Upon detection of sensor attacks/faults and sensor reconfiguration, we want to improve state estimation performance to accommodate the updated sensor measurement model.
\begin{problem} \label{problem2}
\textit{(Estimation Error Minimization)} Create a policy $\mathcal{P}$ where an agent $i \hspace{-1pt} \in \hspace{-1pt} \mathcal{V}$ adaptively estimates the unknown covariance $\bar{\bm{\mathrm{R}}}_i^{(k)}$ for the RSSI-based position measurements that is robust to an unreliable on-board state estimate. Given the updated sensor model in Problem \ref{problem1}, the policy $\mathcal{P}$ follows:
\begin{equation} \label{prob:minimize_er}
    \mathcal{P}\big( \bar{\bm{\mathrm{R}}}_i^{(k)} \big) \longrightarrow \min \Big( \big( \bm{\mathrm{e}}_i^{(k)} \big)^{\mathsf{T}} \bm{\mathrm{e}}_i^{(k)} \Big)
\end{equation}
to minimize its state estimation error $\bm{\mathrm{e}}_i^{(k)} \hspace{-.4pt} = \bm{\mathrm{x}}_i^{(k)} \hspace{-.4pt} - \hat{\bm{\mathrm{x}}}_i^{(k)}$ to provide robust estimation performance within MAS formations.
\end{problem}

\section{Framework} \label{sec:framework}

In this section, we describe the decentralized framework for detection of cyber attacks and faults to on-board positioning sensors and the recovery method by utilizing RSSI measurements from nearby agents (i.e., mobile landmarks) in the MAS to replace the compromised on-board sensor. The block diagram in Fig.~\ref{fig:architecture} summarizes the proposed framework followed by each agent in the swarm to recover from localization sensor/fault to maintain control performance within the formation. As an agent $i$ discovers anomalous behavior to its localization sensors, it switches to a recovery mode which relies on noisy RSSI measurements from nearby uncompromised agents to replace the unreliable on-board sensor measurements. Then, the agent adaptively updates its KF to accommodate the unknown RSSI-based position measurement covariance for improved control performance.
\begin{figure}[ht!]
\vspace{-5pt}
\centering
\includegraphics[width=0.44\textwidth]{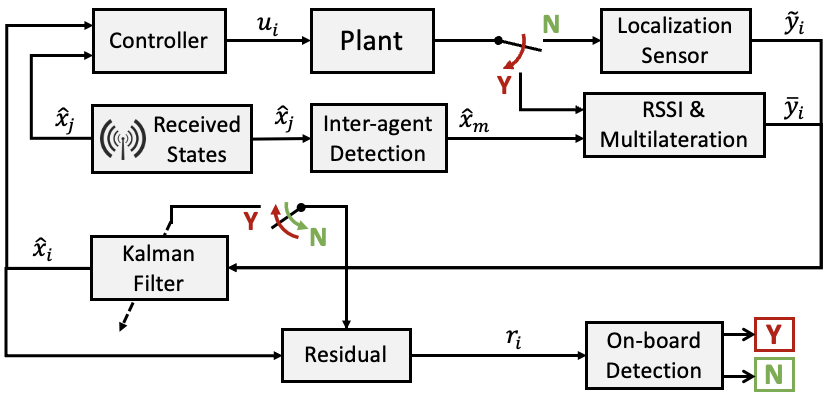}
\vspace{-7pt}
\caption{Overall control architecture followed by each agent $i \in \mathcal{V}$ to remain resilient from localization (i.e., position) sensor attacks and/or faults.}
\label{fig:architecture}
\vspace{-13pt}
\end{figure}

\subsection{Anomaly Detection} \label{sec:Detection}

Each agent $i \in \mathcal{V}$ in our proposed framework monitors for inconsistent behavior of both its on-board position sensor and position information from nearby agents. Let us define the {\em{measurement residual}} on an agent $i$:
\begin{equation} \label{eq:residual}
    \bm{\mathrm{r}}_i^{(k)} = \widetilde{\bm{\mathrm{y}}}_{i}^{(k)} - \bm{\mathrm{C}} \hat{\bm{\mathrm{x}}}^{(k|k-1)}_{i}
\end{equation}
as the difference between the measurements and the prediction. We can model the measurement residual covariance matrix (assuming the KF has converged to steady state) during attack-free conditions as $\bm{\Sigma}_i = \bm{\mathrm{C}}\bm{\mathrm{P}}_i^{(\infty)} \bm{\mathrm{C}}^{\mathsf{T}} + \bm{\mathrm{R}}$ where the steady state estimation covariance $\bm{\mathrm{P}}_i^{(\infty)}$ is found from the discrete Riccati equation.
Each agent $i$ monitors its $D$-dimensional position sensor measurements for anomalies by way of the commonly-used chi-squared scheme, which produces a scalar quadratic on-board \textit{test measure} computed by
\begin{equation} \label{eq:chisquare_testmeasure}
    \mathrm{z}_i^{(k)} = \big( \bm{\mathrm{r}}_{i,[1:D]}^{(k)} \big)^{\mathsf{T}} \bar{\bm{\Sigma}}_i^{-1} \hspace{1.5pt} \bm{\mathrm{r}}_{i,[1:D]}^{(k)}
\end{equation}
which has an expected chi-squared distribution $\mathrm{z}_i^{(k)} \sim \chi^2(D)$ with $D$ degrees of freedom. The matrix $\bar{\bm{\Sigma}}_i \in \R^{D \times D}$ represents the position sensor measurement covariance block within $\bm{\Sigma}_i =${\scriptsize \setlength\arraycolsep{2pt} $ \begin{bmatrix} \bar{\bm{\Sigma}}_i & * \\ *^{\mathsf{T}} & \breve{\bm{\Sigma}}_i \end{bmatrix}$}
where $ \breve{\bm{\Sigma}}_i$ represents the non-position covariance block corresponding to the remaining sensors.

Similarly, each agent monitors for expected behavior of nearby agents according to the control consensus model in \eqref{eq:consensus_control} that is followed by all agents. Each agent $i$ receives both state $\hat{\bm{\mathrm{x}}}^{(k)}_j$ and control input $\bm{\mathrm{u}}^{(k)}_j$ information from any agent $j \in \mathcal{V} \setminus \{i\}$. State evolution predictions $\hat{\bm{\mathrm{x}}}_{ij}^{(k+1)} \in \R^n$ for an agent $j$ are given as:
\begin{equation} \label{eq:neighbor_predict}
    \hat{\bm{\mathrm{x}}}_{ij}^{(k+1)} = \bm{\mathrm{A}} \hat{\bm{\mathrm{x}}}_{j}^{(k)} + \bm{\mathrm{B}} \bm{\mathrm{u}}_{j}^{(k)}
\end{equation}
which is computed by an agent $i$. At every time iteration $k \in \N$, an agent $i$ computes the \textit{inter-agent residual} \cite{Paul_TRO},
\begin{equation} \label{eq:prediction_residual}
    \bm{\mathrm{r}}_{ij}^{(k)} = \hat{\bm{\mathrm{x}}}_j^{(k)} - \hat{\bm{\mathrm{x}}}_{ij}^{(k)} \in \R^n.
\end{equation}
to monitor for consistent behavior of nearby agents.

If the agent $j$ is behaving in a nominal fashion, each inter-agent residual element $q = 1,\dots,n$ follows the distribution $\mathrm{r}_{ij,q}^{(k)} \sim \mathcal{N}(0,\sigma_{ij,q}^2)$ given $\sigma_{ij,q}^2 = \sum_{s=1}^{N_{\text{s}}} \big( \mathrm{K}_{(q,s)} \sigma_{j,s} \big)^2$ where
$\mathrm{K}_{(q,s)}$ represents the element of the $q$th row and $s$th column of the Kalman gain $\bm{\mathrm{K}}$, and $\sigma^2_{j,s}$ is the $s$th diagonal element in the known measurement covariance matrix $\bm{\Sigma}_j = \bm{\Sigma}_i$ \cite{Paul_TRO}. For ease, we construct the inter-agent covariance matrix $\bm{\Sigma}_{ij} \in \R^{n \times n}$ with diagonal elements equal to $\sigma^2_{ij,q}$, i.e. $\bm{\Sigma}_{ij} = \mathrm{diag}(\sigma^2_{ij,1},\dots,\sigma^2_{ij,n})$. In a similar fashion to the on-board test measure in \eqref{eq:chisquare_testmeasure}, the \textit{inter-agent test measure} is computed as 
\begin{equation} \label{eq:interagent_testmeasure}
    \mathrm{z}_{ij}^{(k)} = \big( \bm{\mathrm{r}}_{ij,[1:D]}^{(k)} \big)^{\mathsf{T}} \bar{\bm{\Sigma}}_{ij}^{-1} \hspace{1.5pt} \bm{\mathrm{r}}_{ij,[1:D]}^{(k)}
\end{equation}
by an agent $i$ to monitor for consistency of an agent $j$ where $\bar{\bm{\Sigma}}_{ij} \in \R^{D \times D}$ represents the inter-agent residual covariance block for position within $\bm{\Sigma}_{ij} =${\scriptsize \setlength\arraycolsep{2pt} $ \begin{bmatrix} \bar{\bm{\Sigma}}_{ij} & * \\[1pt] *^{\mathsf{T}} & \breve{\bm{\Sigma}}_{ij} \end{bmatrix}$}.

To monitor for expected behavior of either test measure in \eqref{eq:chisquare_testmeasure} and \eqref{eq:interagent_testmeasure}, simply denoted as $\mathrm{z}^{(k)}$, an agent $i$ creates an alarm-based mechanism when the test measure exceeds a user-defined threshold $\tau$ which follows:
\begin{equation} \label{eq:thresholding_alarm}
\begin{split}
    \zeta^{(k)} = \left\{ \begin{array}{ll}
	1, & \text{if } \mathrm{z}^{(k)} > \tau, \\
    0, & \text{if } \mathrm{z}^{(k)} \leq \tau.
    \end{array} \right.
\end{split}
\end{equation} 

The threshold parameter $\tau \in \R_{>0}$ is tuned to satisfy a user-defined desired false alarm rate $a^{\mathrm{des}} \in (0,1)$.
\begin{lemma}[Threshold] \label{prop:BD_threshold}
    Let us assume that the sensors on agent $i$ are free of cyber attacks and faults while considering the alarm procedure in \eqref{eq:thresholding_alarm} for the test measure $\mathrm{z}^{(k)} \sim \chi^2(D)$ with a threshold $\tau \in \R_{>0}$. To tune for a desired false alarm rate $a^{\mathrm{des}} \in (0,1)$, the threshold $\tau$ is found by
\begin{equation} \label{eq:BD_threshold}
    \tau = 2\Gamma^{-1} \Big( 1 - a^{\mathrm{des}}, \frac{D}{2} \Big)
\end{equation}
to achieve a desired alarm rate, where $\Gamma^{-1}(\cdot,\cdot)$ is the \textit{inverse regularized lower incomplete gamma function} \cite{statsbook}.
\end{lemma}

Formally, the probability $\PP(\cdot)$ that the test measure exceeds the user-defined threshold is described as $\PP(\mathrm{z}^{(k)} > \tau) \approx a^{\mathrm{des}}$. Each agent $i$ updates its alarm rate estimate $\hat{a}^{(k)}$ with the runtime update $\hat{a}^{(k)} = \hat{a}^{(k-1)} + \frac{\zeta^{(k)} - \hat{a}^{(k-1)}}{\ell}$
where $\hat{a}^{(0)} = a^{\mathrm{des}}$, $\ell > \N_+$, and the alarm rate estimate can be approximated to a Normal distribution with a variance that shares properties of the exponential moving average \cite{moving_average}.

\begin{corollary}[Detection Bounds] \label{cor:detection_bounds}
    Given the tuned threshold \eqref{eq:BD_threshold} for a desired false alarm rate $a^{\mathrm{des}} \in (0,1)$, the position sensor measurements are behaving as expected with a level of significance $\alpha \in (0,1)$ if the estimated alarm rate $\hat{a}^{(k)} \in [0,1]$ satisfies the detection bounds $\hat{a}^{(k)} \in [\Tau_-, \Tau_+]$.
\end{corollary}

\begin{proof}
    We construct confidence intervals with a level of significance $\alpha \in (0,1)$ for a Normally distributed random variable with accompanying z-score $\big| \Phi^{-1} \big( \frac{\alpha}{2} \big) \big|$, that provide the detection bounds
    \begin{equation} \label{eq:BD_magnitude_bounds2}
        \Tau_{\pm} = a^{\mathrm{des}} \pm \Big| \Phi^{-1} \Big( \frac{\alpha}{2} \Big) \Big| \sqrt{\frac{ a^{\mathrm{des}} (1-a^{\mathrm{des}})}{2\ell -1}}
    \end{equation}
    where alarm rate estimates that go beyond these bounds are exhibiting anomalous behavior, thus concluding the proof.
\end{proof}

To summarize Corollary \ref{cor:detection_bounds}, when the estimated alarm rate $\hat{a}^{(k)}_i$ for detection of inconsistencies to the on-board test measure $\mathrm{z}_i^{(k)}$ no longer satisfies the detection bounds $\hat{a}_i^{(k)} \not\in [\Tau_-, \Tau_+]$,
agent $i$ detects that a cyber attack or fault to its positioning sensors is present.
In the case of inter-agent monitoring where $\hat{a}_{ij}^{(k)} \not\in [\Tau_-, \Tau_+]$, an agent $i$ deems an agent $j$ compromised and places the agent into a compromised agent set $\mathcal{V}_i^C \subset \mathcal{V}$, i.e., $j \in \mathcal{V}_i^C$.

\subsection{RSSI-based Position Measurements} \label{sec:Multilateration}

In this subsection, we discuss our localization method that leverages the modeled communication channel described in Section~\ref{sec:Communication_model} to replace the compromised/faulty sensor providing position measurements. Our proposed framework is performed by any compromised agent $i$ that utilizes noisy measured RSSI and received position information from uncompromised neighboring agent's communication broadcasts to compute RSSI-based position measurements. 

Once an agent $i$ detects anomalous position sensor behavior, the agent no longer relies on the on-board position sensor and begins to measure RSSI from any trustworthy nearby agents. An agent $i$ utilizes estimated distances from RSSI measurements to each uncompromised mobile agent $m$ in the set $\mathcal{M}_i \subset \mathcal{C}_i \setminus \{ \mathcal{V}_i^C \cup i \}$ where the set $\mathcal{M}_i = \{ 1, 2, \dots, M \}$ represents the $M$ in-range uncompromised mobile agents (i.e., mobile landmarks). From the known communication model, the observed path loss $PL_{im}^{(k)}$ (in dB) by an agent $i$ of an agent $m \in \mathcal{M}_i$ from RSSI $P_{im,[\text{rx}]}^{(k)}$ (in dBm) is:
\begin{equation} \label{eq:observed_PL}
    PL_{im}^{(k)} = P_{[\text{tx}]} - P_{im,[\text{rx}]}^{(k)}.
\end{equation}

The estimated distance $\hat{\mathrm{d}}_{im}^{(k)}$ based on received (i.e., measured) signal strength of an agent $m$ that is computed by a compromised agent $i$ follows:
\begin{equation} \label{eq:estimated_distance}
    \hat{\mathrm{d}}_{im}^{(k)} = 10^{ \frac{PL_{im}^{(k)} - PL(\mathrm{d}_0)}{10 \beta} }.
\end{equation}

The RSSI-based distance estimate to an agent $m$ are log-normal random variables \cite{tarrio2011weighted}, provided the assumption that the communication channel follows a log-normal shadowing path loss model in \eqref{eq:comm_model}, \eqref{eq:observed_PL}, and \eqref{eq:estimated_distance}. Consequently, the distance estimate in \eqref{eq:estimated_distance} is a biased estimate. We leverage the assumed log-normal distribution to compensate for the biased estimate. The log-normal random variables are described by the parameters $\mu_{\mathrm{d}}$ and $\sigma_{\mathrm{d}}$ \cite{tarrio2011weighted}:
\begin{equation} \label{eq:lognormal_parameters}
    \mu_{\mathrm{d}} = \ln \mathrm{d}_{im}^{(k)}, \;\;\;\;\; \sigma_{\mathrm{d}} = \frac{\sigma_{\Lambda} \ln 10}{10 \beta}
\end{equation}
with an RSSI distance estimate expectation that follows
\begin{equation} \label{eq:exp_RSSI_meas}
    \E \big[\hat{\mathrm{d}}_{im}^{(k)} \big] = \exp \Big\{\mu_{\mathrm{d}} + \frac{\sigma_{\mathrm{d}}^2}{2} \Big\}
\end{equation}
and the expected estimation bias
\begin{equation} \label{eq:error_bias}
    \mathrm{e}_{im,\mathrm{e}}^{(k)} = \E \big[\hat{\mathrm{d}}_{im}^{(k)} \big] - \mathrm{d}_{im}^{(k)} \in \R_{>0}
\end{equation}
defined as the difference between the RSSI-based distance estimate expectation $\E \big[\hat{\mathrm{d}}_{im}^{(k)} \big]$ and the true distance $\mathrm{d}_{im}^{(k)}$. Given that the true distance $\mathrm{d}_{im}^{(k)}$ is unknown in \eqref{eq:lognormal_parameters}, an agent $i$ leverages the distance $\hat{\mathrm{d}}_{im}^{(k)} = \| \hat{\bm{\mathrm{p}}}_{i}^{(k)} - \hat{\bm{\mathrm{p}}}_{m}^{(k)} \|$ which is the estimated distance between agents $i$ and $m$ from the position estimate of agent $i$ and the received position estimate from an agent $m$. We can rewrite equation \eqref{eq:estimated_distance} to include compensation for the expected estimation bias \eqref{eq:error_bias} by
\begin{equation} \label{eq:estimated_distance2}
    \hat{\mathrm{d}}_{im}^{(k)} = 10^{ \frac{PL_{im}^{(k)} - PL(\mathrm{d}_0)}{10 \beta} } - \mathrm{e}_{im,\mathrm{d}}^{(k)}.
\end{equation}

Each $m$th agent's estimated distances are leveraged with their corresponding received positions $\hat{\bm{\mathrm{p}}}_{m}^{(k)} = [\hat{\mathrm{p}}_{m,\mathrm{x}}^{(k)} \; \hat{\mathrm{p}}_{m,\mathrm{y}}^{(k)}]^{\mathsf{T}}$ to find an optimal position $\bar{\bm{\mathrm{p}}}_i^{(k)} = [\bar{\mathrm{p}}_{i,\mathrm{x}}^{(k)} \; \bar{\mathrm{p}}_{i,\mathrm{y}}^{(k)}]^{\mathsf{T}}$ by minimizing the distance error residuals $\epsilon_m \in \R$ where $ m = 1,\dots,M$ using the following set of equations:
\begin{equation} \label{eq:distances_circles}
\begin{split}
    \left\{ \hspace{-5pt} \begin{array}{c}
	\hspace{-8pt} \hat{\mathrm{d}}_{i1}^{(k)} = \sqrt{ \big(\hat{\mathrm{p}}_{1,\mathrm{x}}^{(k)} - \bar{\mathrm{p}}_{i,\mathrm{x}}^{(k)} \big)^2 \hspace{-1pt} + \big(\hat{\mathrm{p}}_{1,\mathrm{y}}^{(k)} - \bar{\mathrm{p}}_{i,\mathrm{y}}^{(k)} \big)^2 + \epsilon_1} \\[2pt]
    \hspace{-8pt} \hat{\mathrm{d}}_{i2}^{(k)} = \sqrt{ \big(\hat{\mathrm{p}}_{2,\mathrm{x}}^{(k)} - \bar{\mathrm{p}}_{i,\mathrm{x}}^{(k)} \big)^2 \hspace{-1pt} + \big(\hat{\mathrm{p}}_{2,\mathrm{y}}^{(k)} - \bar{\mathrm{p}}_{i,\mathrm{y}}^{(k)} \big)^2 + \epsilon_2} \\[-2.5pt]
    \vdots \\
    \hat{\mathrm{d}}_{iM}^{(k)} = \sqrt{ \big(\hat{\mathrm{p}}_{M,\mathrm{x}}^{(k)} - \bar{\mathrm{p}}_{i,\mathrm{x}}^{(k)} \big)^2 \hspace{-1pt} + \big(\hat{\mathrm{p}}_{M,\mathrm{y}}^{(k)} - \bar{\mathrm{p}}_{i,\mathrm{y}}^{(k)} \big)^2 + \epsilon_M}
    \end{array} \right.
\end{split}
\end{equation} 

Next, we subtract the first equation from the remaining equations to obtain a system of $M-1$ linear equations
\vspace{-.5pt}
\begin{equation} \label{eq:linear_equations}
    \bm{\Omega}_i \bar{\bm{\mathrm{p}}}_{i}^{(k)} = \bm{\phi}_i + \bm{\varepsilon}_i
\end{equation}
with $\bm{\Omega}_i \hspace{-.5pt} \in \hspace{-.5pt} \R^{(M-1) \times D}\hspace{-.5pt}$, $\bm{\phi}_i \hspace{-.5pt} \in \hspace{-.5pt} \R^{M-1} \hspace{-.5pt}$, and $\bm{\varepsilon}_i \hspace{-.5pt} \in \hspace{-.5pt} \R^{M-1} \hspace{-.5pt}$ found by:
\begin{equation} \label{eq:Omega}
    \hspace{-3pt} \bm{\Omega}_i = \begin{bmatrix} 2\big(\hat{\mathrm{p}}_{2,\mathrm{x}}^{(k)} - \hat{\mathrm{p}}_{1,\mathrm{x}}^{(k)} \big) & 2\big(\hat{\mathrm{p}}_{2,\mathrm{y}}^{(k)} - \hat{\mathrm{p}}_{1,\mathrm{y}}^{(k)} \big) \\[2pt] 2\big(\hat{\mathrm{p}}_{3,\mathrm{x}}^{(k)} - \hat{\mathrm{p}}_{1,\mathrm{x}}^{(k)} \big) & 2 \big(\hat{\mathrm{p}}_{3,\mathrm{y}}^{(k)} - \hat{\mathrm{p}}_{1,\mathrm{y}}^{(k)} \big) \\[-2.5pt] \vdots & \vdots \\[-1pt] 2 \big(\hat{\mathrm{p}}_{M,\mathrm{x}}^{(k)} - \hat{\mathrm{p}}_{1,\mathrm{x}}^{(k)} \big) & 2\big(\hat{\mathrm{p}}_{M,\mathrm{y}}^{(k)} - \hat{\mathrm{p}}_{1,\mathrm{y}}^{(k)} \big) \\ \end{bmatrix}
\end{equation}
\vspace{-1pt}
\begin{equation} \label{eq:phi}
    \hspace{-9pt} \bm{\phi}_i = \begin{bmatrix} \big(\hat{\mathrm{d}}_{i2}^{(k)}\big)^2 - \big(\hat{\mathrm{d}}_{i1}^{(k)}\big)^2 + \mathrm{b}_{2}^{(k)} - \mathrm{b}_{1}^{(k)} \\[2pt] \big(\hat{\mathrm{d}}_{i3}^{(k)}\big)^2 - \big(\hat{\mathrm{d}}_{i1}^{(k)}\big)^2 + \mathrm{b}_{3}^{(k)} - \mathrm{b}_{1}^{(k)} \\[-2pt] \vdots \\[-1pt] \big(\hat{\mathrm{d}}_{iM}^{(k)}\big)^2 - \big(\hat{\mathrm{d}}_{i1}^{(k)}\big)^2 + \mathrm{b}_{M}^{(k)} - \mathrm{b}_{1}^{(k)} \\ \end{bmatrix}
\end{equation}
\vspace{-1pt}
\begin{equation} \label{eq:epsilon}
    \hspace{-3pt} \bm{\varepsilon}_i = \begin{bmatrix} \epsilon_{2}-\epsilon_{1} & \epsilon_{3}-\epsilon_{1}  & \dots & \epsilon_{M}-\epsilon_{1} \end{bmatrix}^{\mathsf{T}}
\end{equation}
where $\mathrm{b}_{m}^{(k)} = \big(\hat{\mathrm{p}}_{m,\mathrm{x}}^{(k)}\big)^2 + \big( \hat{\mathrm{p}}_{m,\mathrm{y}}^{(k)}\big)^2$, $\forall m \in \mathcal{M}_i$. To optimize the position $\bar{\bm{\mathrm{p}}}_i^{(k)} \in \R^D$, we use the following objective function:
\begin{equation} \label{eq:sum_of_squares}
    J \big( \bar{\bm{\mathrm{p}}}_{i}^{(k)} \big) = {\arg \min} \bigg[ \Big\| \bm{\mathrm{W}}_i^{-\frac{1}{2}} \big( \bm{\Omega}_i \bar{\bm{\mathrm{p}}}_{i}^{(k)} - \bm{\phi}_i \big) \Big\|^2 \bigg]
\end{equation}
where $\bm{\mathrm{W}}_i \hspace{-.3pt} \in \hspace{-.3pt} \R^{(M-1) \times (M-1)}$ is a weighting matrix that minimizes the sum of squares of the distance error residual vector $\bm{\varepsilon}_i$. The optimal position is found by solving the objective function \eqref{eq:sum_of_squares} with a weighted least squares (WLS) estimator:
\begin{equation} \label{eq:WLS}
    \bar{\bm{\mathrm{y}}}_{i,[1:D]}^{(k)} = \bar{\bm{\mathrm{p}}}_i^{(k)} = ( \bm{\Omega}_i^{\mathsf{T}} \bm{\mathrm{W}}_i^{-1} \bm{\Omega}_i )^{-1} \bm{\Omega}_i^{\mathsf{T}} \bm{\mathrm{W}}_i^{-1} \bm{\phi}_i.
\end{equation}

This optimal position \eqref{eq:WLS} is equal to the compromised agent $i$'s RSSI-based position measurement $\bar{\bm{\mathrm{y}}}_{i,[1:D]}^{(k)}$.

\begin{proposition}
    An agent $i$ can compute a feasible solution for an RSSI-based position measurement \eqref{eq:WLS} if the number of nearby uncompromised mobile agents $M$ satisfies $M \geq D + 1$.
\end{proposition}

\begin{proof}
    Given the dimensions of the matrix $\bm{\Omega}_i \hspace{-.5pt} \in \hspace{-.5pt} \R^{(M-1) \times D}\hspace{-.5pt}$ and the resulting estimated position vector $\bar{\bm{\mathrm{p}}}_i^{(k)} \in \R^D$ within \eqref{eq:WLS}, we have a trivial proof, as under-determined systems have an infinite number of solutions due to linear dependency, i.e., $\mathrm{rank} \hspace{1pt}(\bm{\Omega}_i) < D$.
\end{proof}


Next, the weighting matrix $\bm{\mathrm{W}}_i$ is characterized to compute an optimal RSSI-based measurement error from \eqref{eq:WLS}.



\subsection{Weighting Matrix} \label{sec:Hyperbolic}

To compute the weighting matrix, we first assume that the RSSI-based estimated distances $\hat{\mathrm{d}}_{im}^{(k)}$ from each uncompromised agent $m \in \mathcal{M}_i$ are independent from each other. Thus, the weighting matrix $\bm{\mathrm{W}}_i$ is the covariance of the vector $\bm{\phi}_i$ that can be calculated by a hyperbolic weighting matrix with each $(p_{\mathrm{w}},q_{\mathrm{w}})$th element (i.e., $\mathrm{w}_{p_{\mathrm{w}}q_{\mathrm{w}}}$) given as \cite{tarrio2011weighted}:
\begin{equation} \label{eq:weighting_matrix}
    \mathrm{w}_{p_{\mathrm{w}}q_{\mathrm{w}}} \hspace{-2pt}=\hspace{-1pt} \left\{ \hspace{-4pt} \begin{array}{ll}
	\V \hspace{-.5pt} \Big[ \hspace{-1.5pt}\big( \hat{\mathrm{d}}_{i1}^{(k)} \big)^2 \hspace{-1pt}\Big] \hspace{-1.3pt} + \hspace{-1pt} \V \hspace{-.5pt} \Big[ \hspace{-1.5pt}\big( \hat{\mathrm{d}}_{i{(p_{\mathrm{w}}+1)}}^{(k)} \big)^2 \hspace{-1pt}\Big] & \hspace{-3.5pt}\text{if } p_{\mathrm{w}} \hspace{-1pt} = \hspace{-1pt} q_{\mathrm{w}} \\[5pt]
    \V \hspace{-.5pt} \Big[ \hspace{-1.5pt}\big(\hat{\mathrm{d}}_{i1}^{(k)} \big)^2 \hspace{-1pt}\Big] & \hspace{-3.5pt}\text{if } p_{\mathrm{w}} \hspace{-1pt} \ne \hspace{-1pt} q_{\mathrm{w}}
    \end{array} \right.
\end{equation}
where the elements of the weighting matrix are inter-agent RSSI-based estimation variances $\V[\cdot]$. Given the log-normal shadowing path loss communication model in \eqref{eq:comm_model}, \eqref{eq:observed_PL}, and \eqref{eq:estimated_distance}, the RSSI-based estimated distances $\hat{\mathrm{d}}_{im}^{(k)}$ in \eqref{eq:estimated_distance2} are log-normal random variables where we can leverage the modeled variances. The estimation variances are computed by \cite{tarrio2011weighted}:
\begin{equation} \label{eq:estimated_variance}
    \V \big[ \big( \hat{\mathrm{d}}_{im}^{(k)} \big)^2 \big] = \E\big[ \big( \hat{\mathrm{d}}_{im}^{(k)} \big)^4 \big] - \big( \E \big[ \big( \hat{\mathrm{d}}_{im}^{(k)} \big)^2 \big] \big)^2
\end{equation}
which are found from the $c$th moment of a given log-normal distributed variable that is represented as $\E[(\hat{\mathrm{d}}_{im})^c] = \exp \big\{ {c\mu_{\mathrm{d}} + \frac{c^2\sigma^2_{\mathrm{d}}}{2}} \big\}$ with log-normal parameters. 
Therefore, the second and fourth moments from \eqref{eq:estimated_variance} are:
\begin{align} \label{eq:cth_moment2}
    \E\big[ \big( \hat{\mathrm{d}}_{im}^{(k)} \big)^2 \big] &= \exp \big\{ {2\mu_{\mathrm{d}} + 2\sigma^2_{\mathrm{d}}} \big\}, \\ 
    \label{eq:cth_moment4} \E \big[ \big( \hat{\mathrm{d}}_{im}^{(k)} \big)^4 \big] &= \exp \big\{ {4\mu_{\mathrm{d}} + 8\sigma^2_{\mathrm{d}}} \big\},
\end{align}
then substituting \eqref{eq:cth_moment2} and \eqref{eq:cth_moment4} into \eqref{eq:estimated_variance}, we result in
\begin{equation} \label{eq:covariance_elements}
    \V \big[ \big( \hat{\mathrm{d}}_{im}^{(k)} \big)^2 \big] = \exp \Big\{ {4\mu_{\mathrm{d}}} \Big( \exp \big\{ {8\sigma_{\mathrm{d}}^2} \big\} - \exp \big\{ {4\sigma_{\mathrm{d}}^2} \big\} \Big) \Big\}
\end{equation}
which are used in the construction of the weighting matrix \eqref{eq:weighting_matrix} to compute the optimized RSSI-based position measurement $\bar{\bm{\mathrm{y}}}_{i,[1:D]}^{(k)}$ by using the WLS algorithm \eqref{eq:WLS}.

\subsection{Robust Adaptive Noise Covariance Estimation} \label{sec:Adaptive_KF}

With the sensor model updated to replace nominal position sensor with RSSI-based position measurements, the measurement covariance matrix $\bm{\mathrm{R}}$ used during nominal conditions is no longer suitable for optimal state estimation. A compromised agent $i$ leverages the recursive KF process to allow for adaptive updates (i.e., time varying) of the measurement noise covariance matrix $\bm{\mathrm{R}}_i^{(k)}$. The adaptive KF can be described in three phases, as follows:

\subsubsection{Prediction} As in the same manner of a typical recursive KF, the following equations denote
\begin{align} \label{eq:state_predict}
	\hat{\bm{\mathrm{x}}}_{i}^{(k|k-1)} &= \bm{\mathrm{A}} \hat{\bm{\mathrm{x}}}^{(k-1|k-1)}_{i} + \bm{\mathrm{B}} \bm{\mathrm{u}}^{(k-1)}_{i}, \\
    \label{eq:covariance_predict}
	\bm{\mathrm{P}}^{(k|k-1)}_{i} &= \bm{\mathrm{A}} \bm{\mathrm{P}}^{(k-1|k-1)}_{i} \bm{\mathrm{A}}^{\mathsf{T}} + \bm{\mathrm{Q}},
\end{align}
the predicted state and estimation error covariance.
	
\subsubsection{Correction} 
In this work, we introduce a novel residual-based method to estimate the unknown covariance of the RSSI-based measurements. Due to attacks/faults to position sensors, the residual employed for estimation must omit the use of the unreliable position estimate $\hat{\bm{\mathrm{x}}}_{i,[1:D]}^{(k)}$; hence the commonly-used measurement residual vector \eqref{eq:residual} can not be utilized (such as in \cite{AdaptiveR_EKF,AdaptiveR_UKF}). To deal with this problem, we leverage the RSSI-based position measurements \eqref{eq:WLS}, system model \eqref{eq:dynamical_model} and \eqref{eq:output_vector}, and the remaining states in the state estimation vector $\hat{\bm{\mathrm{x}}}_{i,[(D+1):n]}^{(k)}$ to estimate the RSSI-based position measurement covariance in a robust manner. Furthermore, we leverage known properties from serial randomness over the sequence of data \cite{Paul_ACC} in measurement covariance estimation.

First, to reconfigure to RSSI-based position measurements, an agent $i$ updates the first $D$ rows of the output matrix $\bm{\mathrm{C}}$ to form the updated output matrix $\bar{\bm{\mathrm{C}}}_i$ by
\begin{equation} \label{eq:pos_outputMatrix}
    \bar{\bm{\mathrm{C}}}_{i,[1:D]} = \Big[ \; \bm{\mathrm{I}}_{D} \;\;\;\; \bm{0}_{D \times (N_{\mathrm{s}}-D)} \; \Big]
\end{equation}
where $\bm{\mathrm{I}}_{D} \hspace{-1pt} \in \hspace{-1pt} \R^{D \times D}$ denotes a $D$-dimension identity matrix, as the RSSI-based position measurements have a $1$:$1$ mapping with the position states (i.e., $\bar{\bm{\mathrm{y}}}_{i,[1:D]}^{(k)} \hspace{-1pt}= \hspace{-1pt} \bar{\bm{\mathrm{p}}}_{i}^{(k)}$ where $\bar{\bm{\mathrm{p}}}_{i}^{(k)}$ is the direct mapping of the state from the RSSI-based position measurement in \eqref{eq:WLS}). The updated output matrix used by the adaptive KF is denoted by $\bar{\bm{\mathrm{C}}}_i \hspace{-.6pt} \in \hspace{-.6pt} \R^{N_{\mathrm{s}} \times n}$ such that the first $D$ rows described in \eqref{eq:pos_outputMatrix} and the remaining $N_{\mathrm{s}}-D$ rows (if applicable) are the same as the output matrix $\bm{\mathrm{C}}$.
\begin{assumption}
    The RSSI-based position measurements can be approximated as a Gaussian distributed vector with covariance $\bar{\bm{\mathrm{R}}}_i^{(k)} \hspace{-1pt} \in \hspace{-1pt} \R^{D \times D}$ centered over the true position $\bm{\mathrm{p}}_i^{(k)}$ of the compromised agent $i$ (i.e., $\bar{\bm{\mathrm{y}}}_{i,[1:D]}^{(k)} \approx \mathcal{N}(\bm{\mathrm{p}}_i^{(k)}, \bar{\bm{\mathrm{R}}}_i^{(k)} ) $).
\end{assumption}

\begin{lemma}
    A compromised agent $i$ that has reconfigured its sensor model to incorporate RSSI-based position measurements $\bar{\bm{\mathrm{y}}}_{i,[1:D]}^{(k)}$ updates its covariance matrix for position measurements $\bar{\bm{\mathrm{R}}}_i^{(k)}$ such that $\bar{\bm{\mathrm{R}}}^{(k)}_{i} = \frac{ \bar{\bm{\Sigma}}^{(k)}_{i} - 2\bar{\bm{\mathrm{Q}}} }{2}$ to minimize its state estimation error.
\end{lemma}

\begin{proof}
We leverage properties of serial randomness to aid in updating the unknown covariance of RSSI-based position measurements (i.e., the observation of consecutive position measurements computed in \eqref{eq:WLS}). To begin, a compromised agent $i$ computes the \textit{RSSI-based measurement residual}
\begin{equation} \label{eq:pos_residual}
    \bar{\bm{\mathrm{r}}}_{i}^{(k)} = \bar{\bm{\mathrm{y}}}_{i,[1:D]}^{(k)} - \bar{\bm{\mathrm{C}}}_{i,[1:D]} \hat{\bar{\bm{\mathrm{x}}}}_{i}^{(k)}
\end{equation}
which is a comparison between the RSSI-based position measurement \eqref{eq:WLS} and a position prediction $\bar{\bm{\mathrm{C}}}_{i,[1:D]} \hat{\bar{\bm{\mathrm{x}}}}_{i}^{(k)}$. The prediction vector $\hat{\bar{\bm{\mathrm{x}}}}_{i}^{(k)}$ in \eqref{eq:pos_residual} is computed by
\begin{equation} \label{eq:meas_prediction}
    \hat{\bar{\bm{\mathrm{x}}}}_{i}^{(k)} = \bm{\mathrm{A}} \bar{\bm{\mathrm{x}}}_{i}^{(k-1)} + \bm{\mathrm{B}} \bm{\mathrm{u}}_i^{(k-1)}
\end{equation}
where $\bar{\bm{\mathrm{x}}}_{i}^{(k-1)} \in \R^n$ represents the previous state estimate vector used for correction at time $k-1$. The vector $\bar{\bm{\mathrm{x}}}_{i}^{(k-1)}$ used for correction is defined as:
\begin{equation} \label{eq:state_predict_update}
    \bar{\bm{\mathrm{x}}}_{i}^{(k-1)} = \big[ \big( \bar{\bm{\mathrm{p}}}_{i}^{(k-1)} \big)^{\mathsf{T}}  \hspace{5pt} \big( \hat{\bm{\mathrm{x}}}_{i,[(D+1):n]}^{(k-1|k-1)} \big)^{\mathsf{T}}  \big]^{\mathsf{T}}.
\end{equation}
which is composed of the state estimate vector with its position elements $\hat{\bm{\mathrm{x}}}_{i,[1:D]}^{(k-1|k-1)}$ replaced with the positions from the RSSI-based position measurement $\bar{\bm{\mathrm{y}}}_{i,[1:D]}^{(k-1)} = \bar{\bm{\mathrm{p}}}_{i}^{(k-1)}$.

At each time $k \in \N$ the difference (i.e., the residual in \eqref{eq:pos_residual}) $\bar{\bm{\mathrm{r}}}_{i}^{(k)} = \bar{\bm{\mathrm{p}}}_{i}^{(k)} - \hat{\bar{\bm{\mathrm{p}}}}_{i}^{(k)}$ where $\hat{\bar{\bm{\mathrm{p}}}}_{i}^{(k)} = \hat{\bar{\bm{\mathrm{x}}}}_{i,[1:D]}^{(k)}$ from \eqref{eq:meas_prediction} is monitored between two consecutive positions at times $k$ and $k-1$. The noise characteristics for the measured positions $\bar{\bm{\mathrm{y}}}_{i,[1:D]}^{(k)}$ are subject to both measurement noise $\bar{\bm{\mathrm{R}}}_i^{(k)}$ and process noise $\bar{\bm{\mathrm{Q}}}$ of the position, where $\bar{\bm{\mathrm{Q}}}$ is the noise covariance block within $\bm{\mathrm{Q}}$ that corresponds to the position states of the compromised sensors. Given these noise properties, the expectation of the RSSI measured residual are described as:
\begin{equation} \label{eq:E_RSSIresidual}
\begin{split}
    \E[\bar{\bm{\mathrm{r}}}_{i}^{(k)}] &= \E[\bm{\eta}_{i,[1:D]}^{(k)}] - \E[\bm{\eta}_{i,[1:D]}^{(k-1)}] \\ 
    & \;\;\; + \E[\bm{\nu}_{i,[1:D]}^{(k-1)}] - \E[\bm{\nu}_{i,[1:D]}^{(k-2)}] = \bm{0},
\end{split}
\end{equation}
\begin{equation} \label{eq:V_RSSIresidual}
\begin{split}
    \hspace{-33pt}\V[\bar{\bm{\mathrm{r}}}_{i,[1:D]}^{(k)}] &= \V[\bm{\eta}_{i,[1:D]}^{(k)}] + \V[\bm{\eta}_{i,[1:D]}^{(k-1)}] \\
    & \;\;\; + \V[\bm{\nu}_{i,[1:D]}^{(k-1)}] + \V[\bm{\nu}_{i,[1:D]}^{(k-2)}] \\
    &= 2 \bar{\bm{\mathrm{R}}}_i^{(k)} + 2\bar{\bm{\mathrm{Q}}} \approx \bar{\bm{\Sigma}}^{(k)}_{i}.
\end{split}
\end{equation}

The objective is to estimate the positive definite covariance matrix of the RSSI measurement residual
\begin{equation} \label{eq:update_cov_residual}
    \hspace{-5pt} \bar{\bm{\Sigma}}^{(k)}_{i} = \E \Big[ \bar{\bm{\mathrm{r}}}_{i}^{(k)} \big( \bar{\bm{\mathrm{r}}}_{i}^{(k)} \big)^{\mathsf{T}} \Big] \in \R^{D \times D}
\end{equation}
by using a rolling runtime estimate similar to \cite{AdaptiveR_EKF}
\begin{equation} \label{eq:update_R2}
    \bar{\bm{\Sigma}}^{(k)}_{i} = (1-\gamma) \bar{\bm{\Sigma}}^{(k-1)}_{i} + \gamma \Big( \bar{\bm{\mathrm{r}}}_{i}^{(k)} \big( \bar{\bm{\mathrm{r}}}_{i}^{(k)} \big)^{\mathsf{T}} \Big)
\end{equation}
where $\gamma \in (0,1)$ is a \textit{forgetting} parameter in adaptively updating the covariance\footnote{We note that a smaller $\gamma$ puts a greater weight on the previous RSSI measurement residual covariance estimate at time $k-1$, thus resulting in less variation in the updated estimate.}. From the approximated estimated covariance in \eqref{eq:V_RSSIresidual}, we can compute the estimated RSSI-based position measurement covariance block
\begin{equation} \label{eq:update_Rbar}
    \bar{\bm{\mathrm{R}}}^{(k)}_{i} = \frac{ \bar{\bm{\Sigma}}^{(k)}_{i} - 2\bar{\bm{\mathrm{Q}}} }{2}
\end{equation}
thus concluding the proof.
\end{proof}

The updated positioning covariance \eqref{eq:update_Rbar} is contained within the updated measurement covariance matrix
\begin{equation} \label{eq:Rpos_covariance}
    \bm{\mathrm{R}}_i^{(k)} = \begin{bmatrix} \bar{\bm{\mathrm{R}}}_i^{(k)} & \bm{0}_{ D \times (N_{\mathrm{s}}-D)} \\[4pt] \bm{0}_{ (N_{\mathrm{s}}-D) \times D} & \breve{\bm{\mathrm{R}}} \end{bmatrix}
\end{equation}
for all $N_{\mathrm{s}}$ sensors, where $\breve{\bm{\mathrm{R}}} \in \R^{(N_{\mathrm{s}}-D)\times (N_{\mathrm{s}}-D)}$ is the measurement covariance of the remaining non-position sensors.
	
\subsubsection{Update} We incorporate the correction phase that adaptively updated the position measurement covariance to optimize the state estimate, we then have the updates to:
\begin{align} \label{eq:Kalman}
	\bm{\mathrm{K}}^{(k)}_{i} &= \bm{\mathrm{P}}^{(k|k-1)}_{i} \bar{\bm{\mathrm{C}}}_i^{\mathsf{T}} \big( \bar{\bm{\mathrm{C}}}_i \bm{\mathrm{P}}^{(k|k-1)}_{i} \bar{\bm{\mathrm{C}}}_i^{\mathsf{T}} + \bm{\mathrm{\mathrm{R}}}_i^{(k)} \big)^{-1}, \\
    \label{eq:updated_estimate}
	 \hat{\bm{\mathrm{x}}}^{(k|k)}_{i} &= \hat{\bm{\mathrm{x}}}^{(k|k-1)}_{i} + \bm{\mathrm{K}}^{(k)}_{i} \big( \bar{\bm{\mathrm{y}}}^{(k)}_{i} - \bar{\bm{\mathrm{C}}}_i\hat{\bm{\mathrm{x}}}^{(k|k-1)}_{i} \big), \\
    \label{eq:updated_covariance}
	 \bm{\mathrm{P}}^{(k|k)}_{i} &= \big( \bm{\mathrm{I}}_n - \bm{\mathrm{K}}^{(k)}_{i} \bar{\bm{\mathrm{C}}}_i \big) \bm{\mathrm{P}}^{(k|k-1)}_{i},
\end{align}
the Kalman gain, state estimate, and estimation covariance.


\section{Results} \label{sec:results}

Our approach is validated with MATLAB simulations on swarms of $N_{\mathrm{a}}=12$ mobile agents modeled with double integrator dynamics satisfying \eqref{eq:dynamical_model}. The MAS performs a go-to-goal operation within a 2-dimensional plane (i.e., $D=2$) while experiencing cyber attacks and faults to on-board positioning sensors. As a case study, we employ a virtual spring-damper mesh \cite{Paul_TRO} for decentralized proximity-based formation control to validate the RSSI-based localization framework for resilient formation control. For all simulations, the agents begin with randomized initial positions and are tasked to perform a go-to-goal operation while maintaining $l^{\mathrm{des}} = 8$m distance from any neighboring agent. Additionally, at time $k=350$, five agents are randomly chosen to suffer from a cyber attack and two experience sensor faults to on-board positioning sensors. The communication model has a path loss exponent $\beta = 2$ and shadowing noise of $\Lambda = \mathcal{N}(0,2)$, while the forgetting parameter used is $\gamma = 0.01$ for measurement covariance estimation.

\begin{figure}[tb!]
\vspace{-1pt}
\hspace{-3pt} \subfloat[\label{fig:first_sim} ]{\setlength{\fboxsep}{0pt}\fbox{\includegraphics[width = 0.238\textwidth]{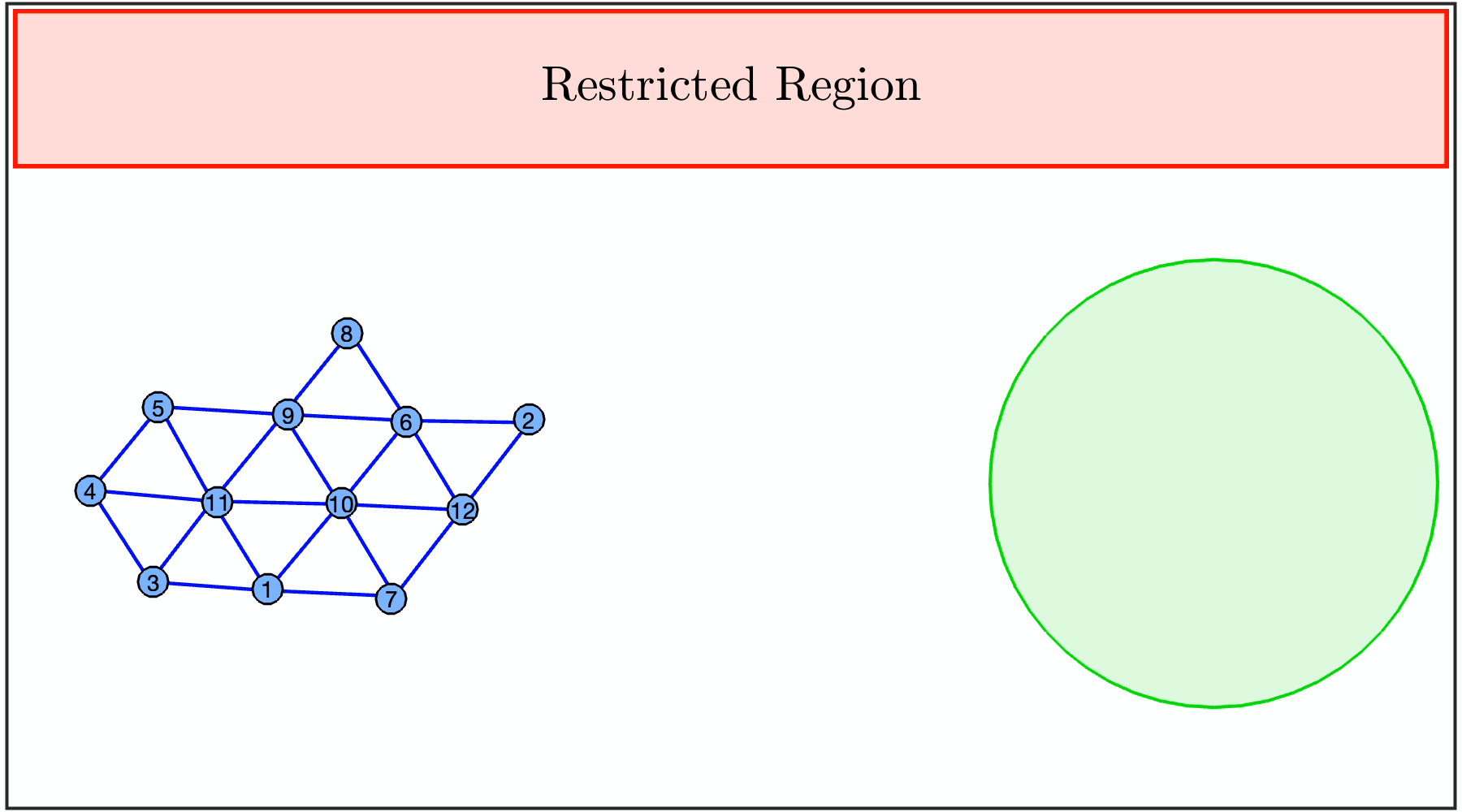}}}
\hspace{-1pt} \subfloat[\label{fig:second_sim} ]{\setlength{\fboxsep}{0pt}\fbox{\includegraphics[width = 0.238\textwidth]{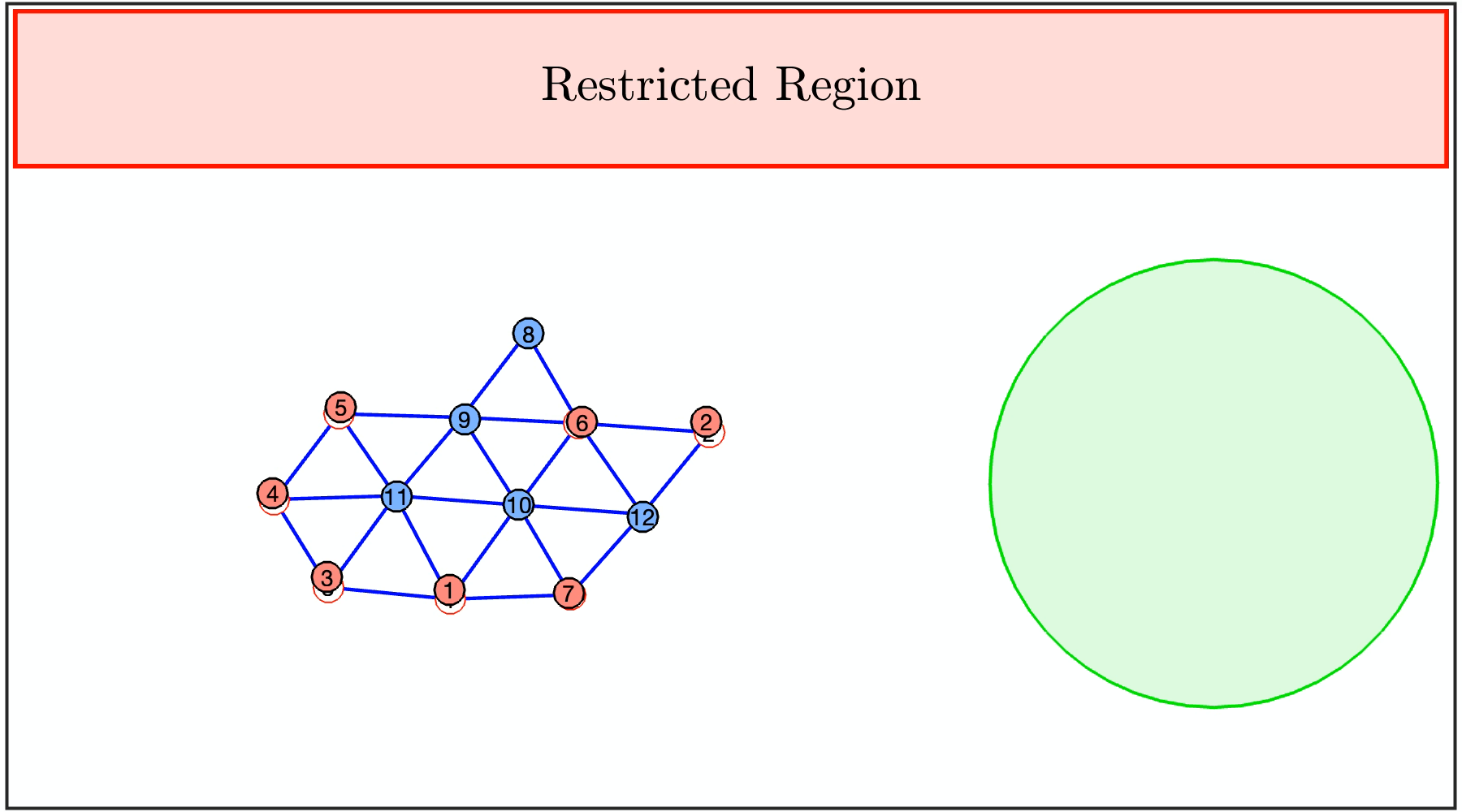}}} \\[3.5pt]
\hspace{-1pt} \subfloat[\label{fig:third_sim} ]{\setlength{\fboxsep}{0pt}\fbox{\includegraphics[width = 0.238\textwidth]{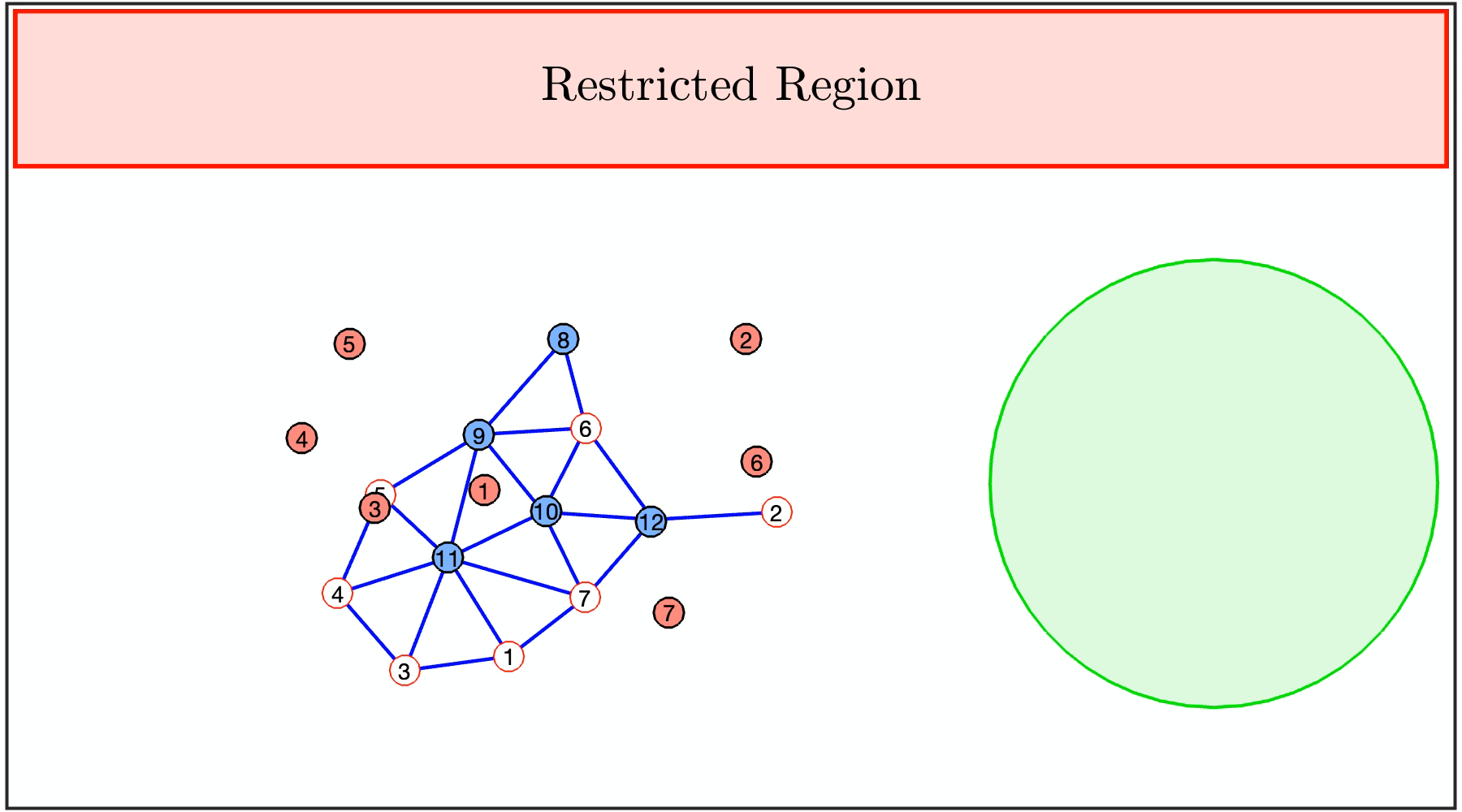}}}
\hspace{-1pt} \subfloat[\label{fig:fourth_sim} ]{\setlength{\fboxsep}{0pt}\fbox{\includegraphics[width = 0.238\textwidth]{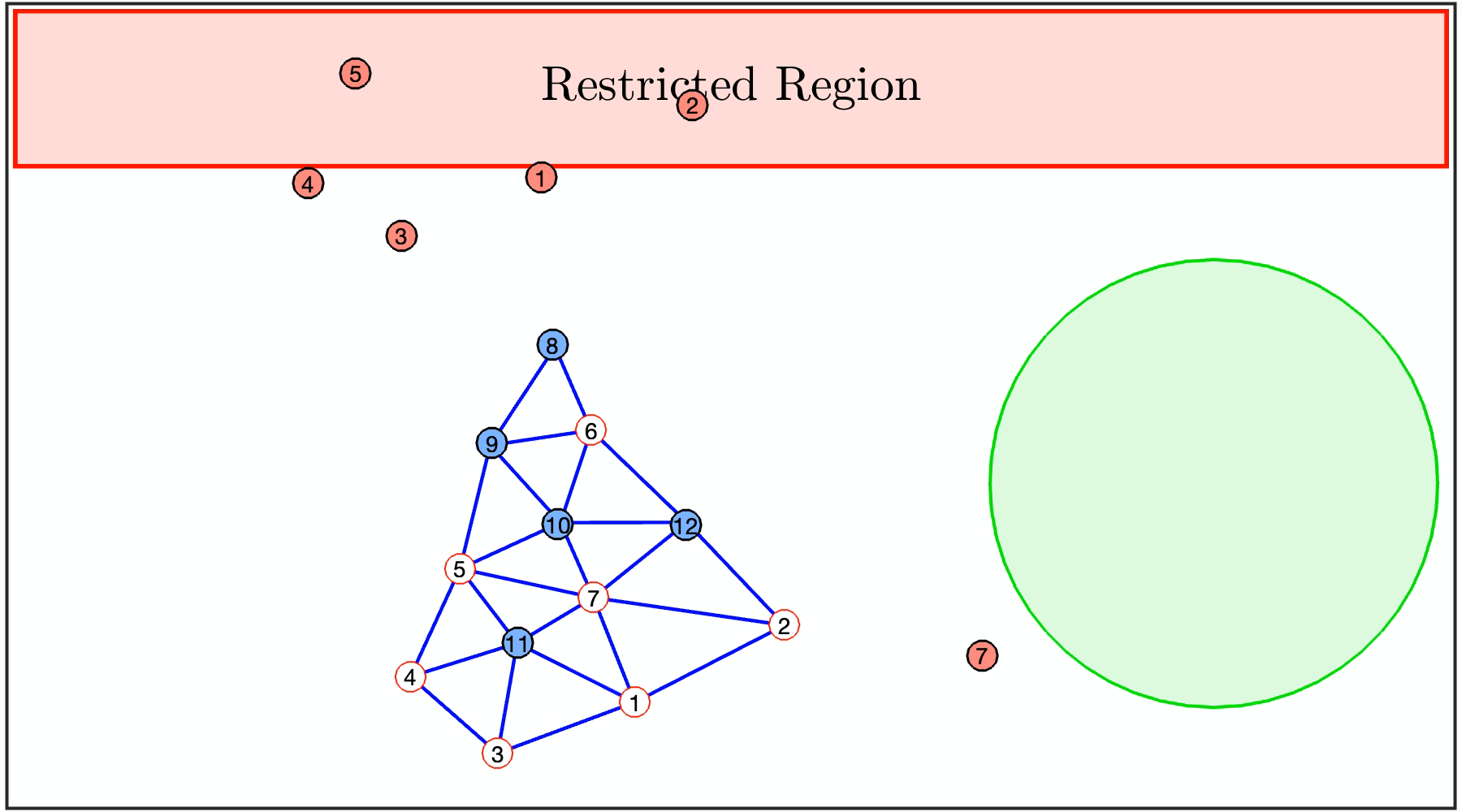}}}
\vspace{-1pt}
\caption{An unprotected system of $N_{\mathrm{a}} = 12$ agents compromised by cyber attacks and faults to on-board position sensors on seven agents (red disks) resulting in them being diverted away from the goal (green region).}
\label{fig:Simulation}
\vspace{-4pt}
\end{figure}

\begin{figure}[tb!]
\hspace{-3pt} \subfloat[\label{fig:first_sim2} ]{\setlength{\fboxsep}{0pt}\fbox{\includegraphics[width = 0.238\textwidth]{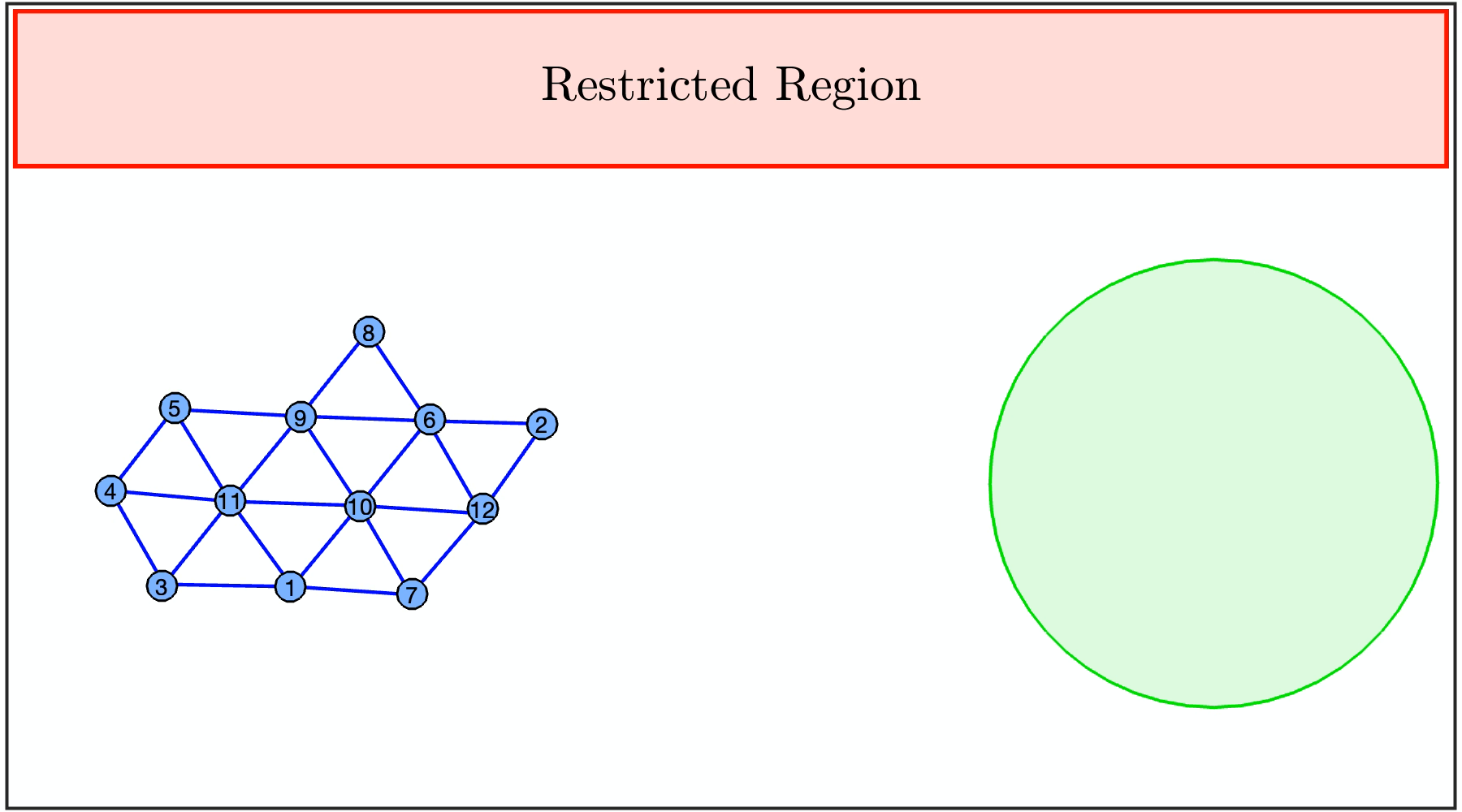}}}
\hspace{-1pt} \subfloat[\label{fig:second_sim2} ]{\setlength{\fboxsep}{0pt}\fbox{\includegraphics[width = 0.238\textwidth]{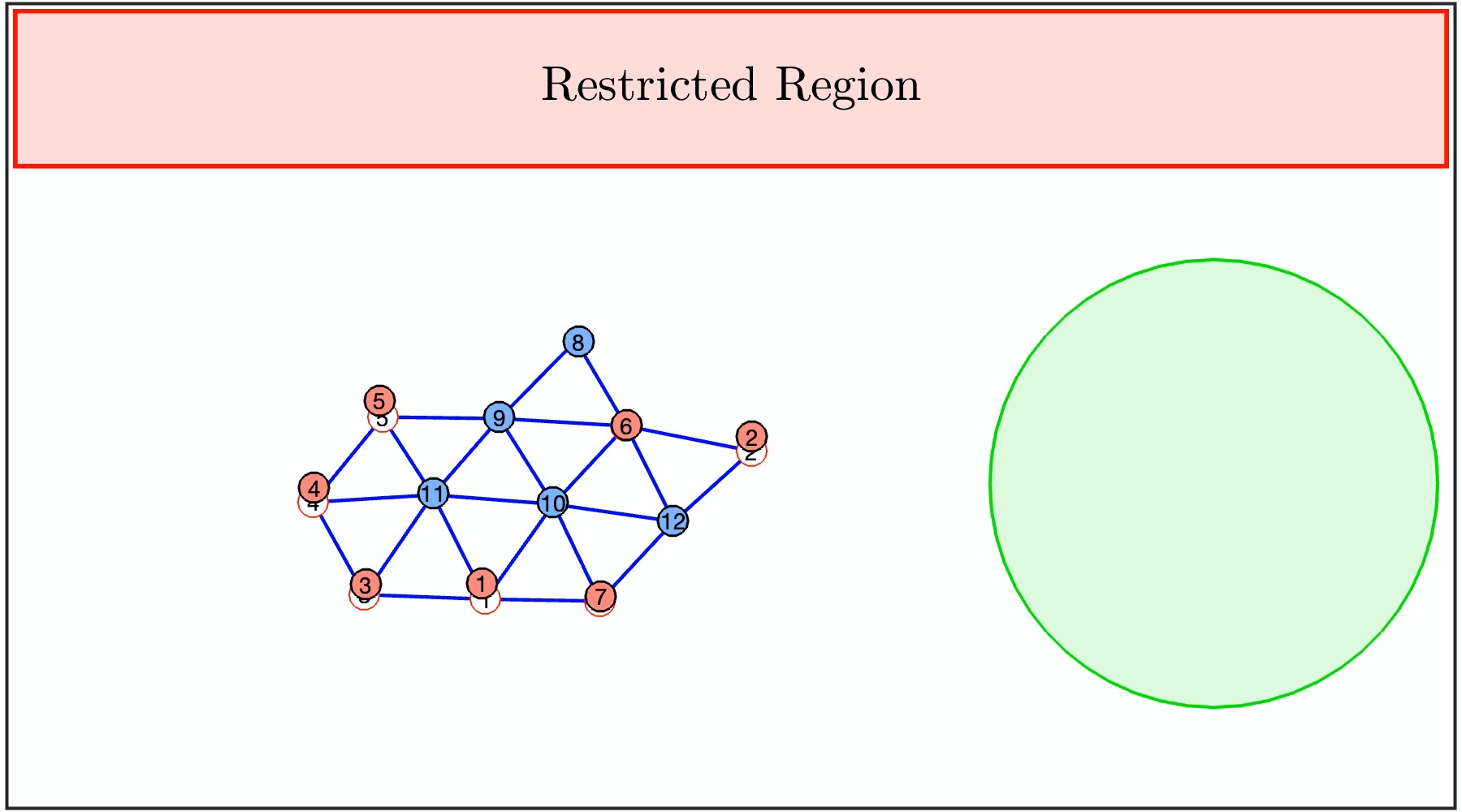}}} \\[3.5pt]
\hspace{-1pt} \subfloat[\label{fig:third_sim2} ]{\setlength{\fboxsep}{0pt}\fbox{\includegraphics[width = 0.238\textwidth]{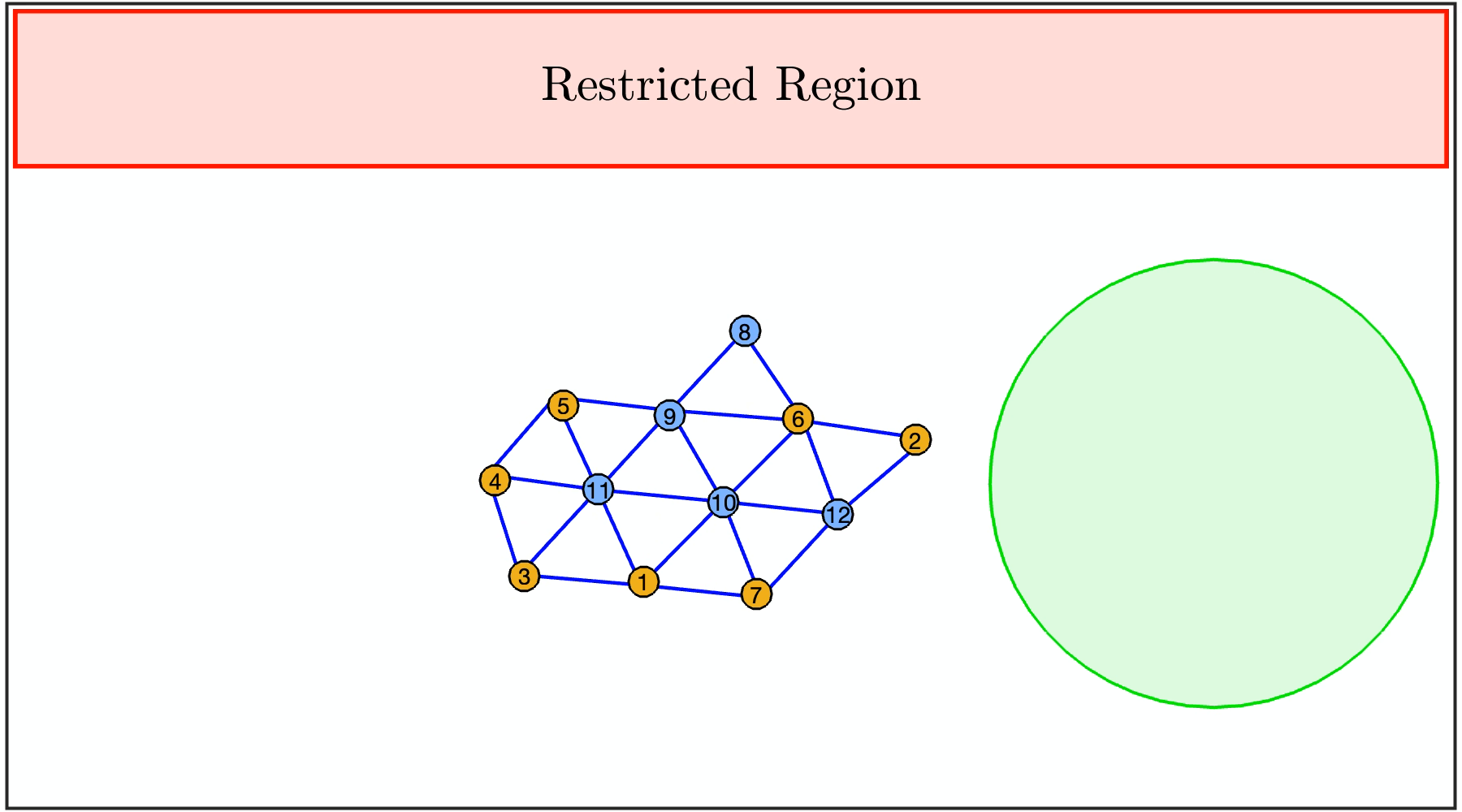}}}
\hspace{-1pt} \subfloat[\label{fig:fourth_sim2} ]{\setlength{\fboxsep}{0pt}\fbox{\includegraphics[width = 0.238\textwidth]{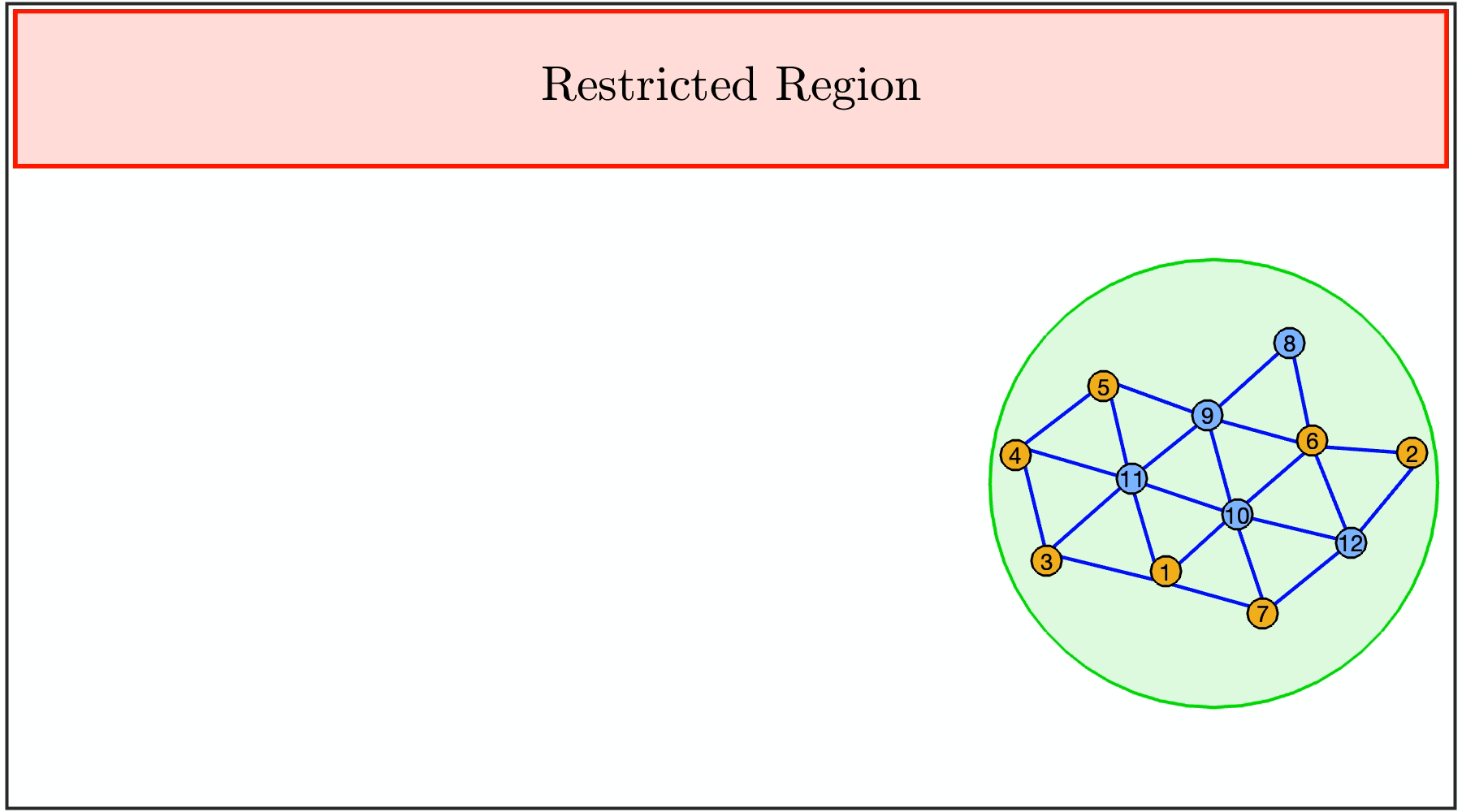}}}
\vspace{-1pt}
\caption{An MAS leveraging our framework resiliently navigates to the desired goal point while experiencing cyber attacks and faults to on-board position sensors (recovered agents in yellow).}
\label{fig:Simulation2}
\vspace{-12pt}
\end{figure}

In the first simulation presented in Fig.~\ref{fig:Simulation}, we show a sequence of snapshots of an unprotected multi-agent system navigating toward a goal region (green disk). The implemented cyber attacks occur simultaneously on the five compromised agents with the intent of diverting their true positions toward the undesired region (red region). Beginning in Fig.~\ref{fig:first_sim}, all agents are uncompromised (blue disks) and perform in a nominal manner; however in Fig.~\ref{fig:second_sim}, the compromised agents are subject to malicious cyber attacks and faults to their positioning sensors. In Figs.~\ref{fig:third_sim}-\subref{fig:fourth_sim}, the true states of the compromised agents (red disks) are driven to undesired regions in the environment, while the remaining uncompromised agents (blue disks) along with the corresponding unreliable (i.e., incorrect) position estimates of the compromised agents (empty disks) continue navigating toward the goal.
In Fig.~\ref{fig:Simulation2} we perform the same simulation as in Fig.~\ref{fig:Simulation}, but this time the MAS is utilizing our framework for resiliency. As shown in the series of snapshots, all seven agents are able to: 1) detect the anomalous on-board positioning sensor behavior and 2) perform sensor reconfiguration to RSSI-based position measurements to resiliently maintain desired performance within the MAS such that all agents safely reach the desired goal.

We provide a comparison between various MAS scenarios by showing the true proximity error between neighboring agents $(i,j) \hspace{-.9pt} \in \hspace{-.9pt}  \mathcal{E}$. We highlight the scenarios, which include: 1) no detection and recovery, 2) a non-robust noise covariance update method of $\bm{\mathrm{R}}_i^{(k)}$ in \cite{AdaptiveR_EKF}, 3) detection and recovery without updating $\bm{\mathrm{R}}_i^{(k)}$, and 4) our proposed robust method for updating $\bm{\mathrm{R}}_i^{(k)}$. Fig.~\ref{fig:FormationError} shows the average formation proximity error over $400$ simulations with randomized initial positions for each scenario. At each time step $k$, the formation proximity error is computed as $\mathrm{E}^{(k)} =  \frac{1}{|\mathcal{E}|} \sum_{\forall (i,j) \in \mathcal{E}} \big| \| \bm{\mathrm{p}}_i^{(k)} - \bm{\mathrm{p}}_j^{(k)} \| - l^{\text{des}} \big| $. In Table~\ref{table:PositionError}, the results of position estimation error $\bm{\mathrm{e}}_{i,x/y}$ in the 2-dimensional coordinate frame are presented for any compromised agent $i$. Our proposed method to adaptively update the unknown RSSI-based position measurement covariance matrix reduces estimation error for more desirable control performance within the swarm.

\begin{figure}[tb!]
\centering
\includegraphics[width = 0.485\textwidth]{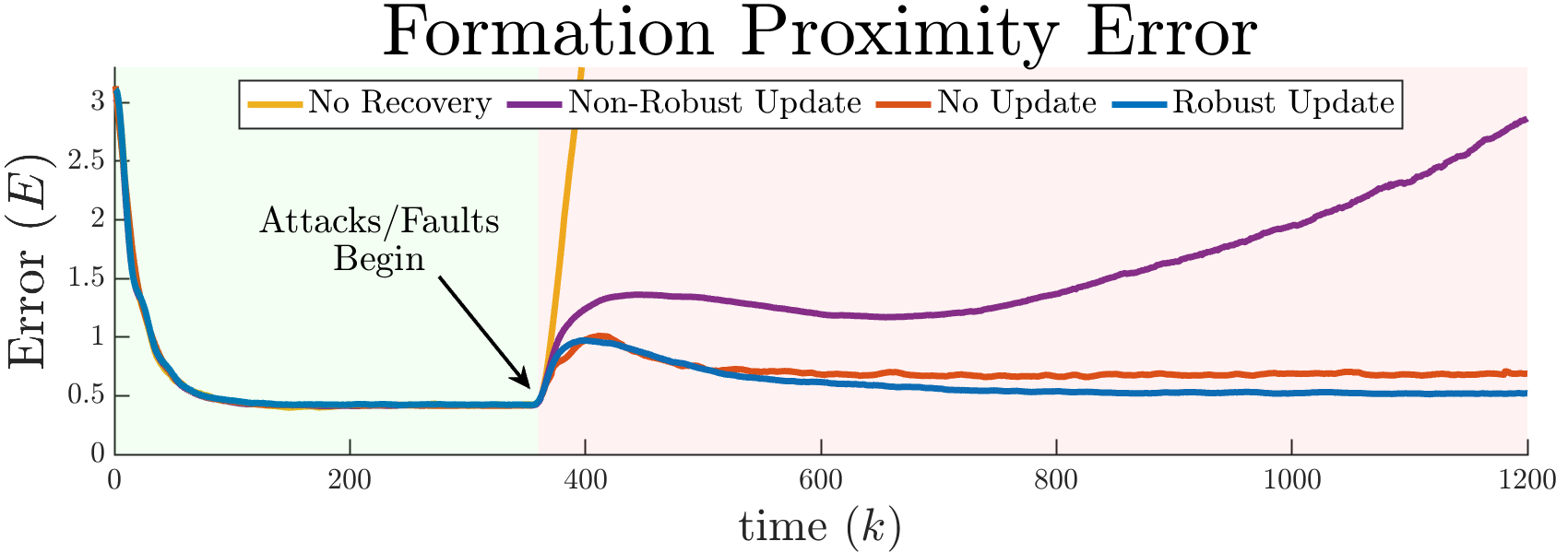}
\vspace{-13pt}
\caption{A comparison of inter-agent proximity error within the formation.}
\label{fig:FormationError}
\vspace{-10pt}
\end{figure}

\begin{table}[htb!]
\caption{Position Estimation Error.}
\vspace{-4pt}
\centering
\begin{tabular}{ p{2.4cm}||p{2cm}|p{2cm} }
 \rule{0pt}{1.2\normalbaselineskip} & \centering No Update & \hspace*{2pt} Robust Update \\[1pt]
 \hline
 \rule{0pt}{.9\normalbaselineskip} \centering Variance $[\mathrm{e}_{i,x/y}]$ & \centering $0.249/0.253$ & \hspace*{3pt} $0.151/0.144$ \\[.2pt]
\end{tabular}
\label{table:PositionError}
\end{table}
\vspace{-4pt}

\section{Conclusions} \label{sec:conclusion}

This paper provides a decentralized framework for multi-agent systems to resiliently navigate in the presence of cyber attacks and/or faults to on-board positioning sensors within open or unknown environments that lack identifiable landmarks (i.e., anchors) and also operate beyond distance/range sensing of other nearby agents. Upon detection of anomalous sensor behavior, an agent performs sensor reconfiguration to leverage RSSI-based measurements from the nearby agents (i.e., mobile landmarks) as a replacement for the original position sensor. An adaptive Kalman Filtering method accommodates the updated position sensor by estimating its unknown measurement covariance to reduce estimation error for improved control performance within the swarm. From here, future work includes improving the robustness of our framework when the assumed path loss model does not hold due to instability of RSSI signals, for example because of the presence of cluttered environments that create multi-path behavior in the communication broadcasts.


\section*{Acknowledgements}

This work is based on research sponsored by the National Science Foundation under grants 1816591 and 1916760. 

\bibliographystyle{IEEEtran}
\linespread{.99}\selectfont
\bibliography{bibliography.bib}

\begin{thebibliography}{10}
\providecommand{\url}[1]{#1}
\csname url@rmstyle\endcsname
\providecommand{\newblock}{\relax}
\providecommand{\bibinfo}[2]{#2}
\providecommand\BIBentrySTDinterwordspacing{\spaceskip=0pt\relax}
\providecommand\BIBentryALTinterwordstretchfactor{4}
\providecommand\BIBentryALTinterwordspacing{\spaceskip=\fontdimen2\font plus
\BIBentryALTinterwordstretchfactor\fontdimen3\font minus
  \fontdimen4\font\relax}
\providecommand\BIBforeignlanguage[2]{{%
\expandafter\ifx\csname l@#1\endcsname\relax
\typeout{** WARNING: IEEEtran.bst: No hyphenation pattern has been}%
\typeout{** loaded for the language `#1'. Using the pattern for}%
\typeout{** the default language instead.}%
\else
\language=\csname l@#1\endcsname
\fi
#2}}

\bibitem{formation_control1}
W.~{Ren}, R.~W. {Beard}, and E.~M. {Atkins}, ``Information consensus in
  multivehicle cooperative control,'' \emph{IEEE Control Systems Magazine},
  vol.~27, no.~2, pp. 71--82, 2007.

\bibitem{7822915}
K.~{Saulnier}, D.~{Saldaña}, A.~{Prorok}, G.~J. {Pappas}, and V.~{Kumar},
  ``Resilient flocking for mobile robot teams,'' \emph{IEEE Robotics and
  Automation Letters}, vol.~2, no.~2, pp. 1039--1046, 2017.

\bibitem{Paul_TRO}
P.~J. Bonczek, R.~Peddi, S.~Gao, and N.~Bezzo, ``Detection of nonrandom
  sign-based behavior for resilient coordination of robotic swarms,''
  \emph{IEEE Transactions on Robotics}, vol.~38, no.~1, pp. 92--109, 2022.

\bibitem{hetergeneous_MAS}
M.~Zamani and A.~P. Aguiar, ``Distributed localization of heterogeneous agents
  with uncertain relative measurements and communications,'' in \emph{54th IEEE
  Conference on Decision and Control}, 2015, pp. 4702--4707.

\bibitem{Nonlinear_MAS}
R.~Maidana, A.~Amory, and A.~Salton, ``Outdoor localization system with
  augmented state extended kalman filter and radio-frequency received signal
  strength,'' in \emph{2019 19th International Conference on Advanced Robotics
  (ICAR)}, 2019, pp. 604--609.

\bibitem{positioning_3D}
J.~Yang, H.~Lee, and K.~Moessner, ``Multilateration localization based on
  singular value decomposition for 3d indoor positioning,'' in \emph{2016
  International Conference on Indoor Positioning and Indoor Navigation (IPIN)},
  2016, pp. 1--8.

\bibitem{MAS_survey}
D.~Zhang, G.~Feng, Y.~Shi, and D.~Srinivasan, ``Physical safety and cyber
  security analysis of multi-agent systems: A survey of recent advances,''
  \emph{IEEE/CAA Journal of Automatica Sinica}, vol.~8, no.~2, pp. 319--333,
  2021.

\bibitem{localization_survey}
F.~Liu, \emph{et~al.}, ``Survey on wifi-based indoor positioning techniques,''
  \emph{IET Communications}, vol.~14, no.~9, pp. 1372--1383, 2020.

\bibitem{survey_LocalizationRSSI}
F.~Zafari, A.~Gkelias, and K.~K. Leung, ``A survey of indoor localization
  systems and technologies,'' \emph{IEEE Communications Surveys Tutorials},
  vol.~21, no.~3, pp. 2568--2599, 2019.

\bibitem{hsu2016particle}
C.-C. Hsu, S.-S. Yeh, and P.-L. Hsu, ``Particle filter design for mobile robot
  localization based on received signal strength indicator,''
  \emph{Transactions of the Institute of Measurement and Control}, vol.~38,
  no.~11, pp. 1311--1319, 2016.

\bibitem{RSSI_DVhop}
J.~Mass-Sanchez, E.~Ruiz-Ibarra, J.~Cortez-Gonz{\'a}lez, A.~Espinoza-Ruiz, and
  L.~A. Castro, ``Weighted hyperbolic dv-hop positioning node localization
  algorithm in wsns,'' \emph{Wireless Personal Communications}, vol.~96, no.~4,
  pp. 5011--5033, 2017.

\bibitem{RSSI_NarrowCorridor}
K.~Hattori, \emph{et~al.}, ``Deployment of wireless mesh network using
  rssi-based swarm robots: Passing narrow corridor by movement function along
  walls,'' \emph{Artif. Life and Robotics}, vol.~21, no.~4, pp. 434--442, 2016.

\bibitem{RSSI_mixed}
L.~Carlino, D.~Jin, M.~Muma, and A.~M. Zoubir, ``Robust distributed cooperative
  rss-based localization for directed graphs in mixed los/nlos environments,''
  \emph{EURASIP Journal on Wireless Communications and Networking}, vol. 2019,
  no.~1, pp. 1--20, 2019.

\bibitem{Mostofi}
C.~R. Karanam, B.~Korany, and Y.~Mostofi, ``Magnitude-based angle-of-arrival
  estimation, localization, and target tracking,'' in \emph{2018 17th ACM/IEEE
  International Conference on Information Processing in Sensor Networks
  (IPSN)}, 2018, pp. 254--265.

\bibitem{oliveira2014rssi}
L.~Oliveira, H.~Li, L.~Almeida, and T.~E. Abrudan, ``Rssi-based relative
  localisation for mobile robots,'' \emph{Ad Hoc Networks}, vol.~13, pp.
  321--335, 2014.

\bibitem{AdaptiveR_EKF}
S.~Akhlaghi, N.~Zhou, and Z.~Huang, ``Adaptive adjustment of noise covariance
  in kalman filter for dynamic state estimation,'' in \emph{2017 IEEE Power
  Energy Society General Meeting}, 2017, pp. 1--5.

\bibitem{AdaptiveR_UKF}
B.~Zheng, P.~Fu, B.~Li, and X.~Yuan, ``A robust adaptive unscented kalman
  filter for nonlinear estimation with uncertain noise covariance,''
  \emph{Sensors}, vol.~18, no.~3, p. 808, 2018.

\bibitem{goldsmith_2005}
A.~Goldsmith, \emph{Wireless Communications}.\hskip 1em plus 0.5em minus
  0.4em\relax Cambridge University Press, 2005.

\bibitem{statsbook}
S.~M. Ross, \emph{Introduction to Probability Models, Ninth Edition}.\hskip 1em
  plus 0.5em minus 0.4em\relax Orlando, FL, USA: Academic Press, Inc., 2006.

\bibitem{moving_average}
S.~I. Gass and C.~M. Harris, ``Encyclopedia of operations research and
  management science,'' \emph{Journal of the Operational Research Society},
  vol.~48, no.~7, pp. 759--760, 1997.

\bibitem{tarrio2011weighted}
P.~Tarr{\'\i}o, A.~M. Bernardos, and J.~R. Casar, ``Weighted least squares
  techniques for improved received signal strength based localization,''
  \emph{Sensors}, vol.~11, no.~9, pp. 8569--8592, 2011.

\bibitem{Paul_ACC}
P.~J. Bonczek and N.~Bezzo, ``Detection of hidden attacks on cyber-physical
  systems from serial magnitude and sign randomness inconsistencies,'' in
  \emph{2021 IEEE American Control Conference (ACC)}, 2021, pp. 3281--3287.

\end{thebibliography}

\end{document}